\documentclass[11pt]{article}
\pdfoutput=1

\usepackage[affil-it]{authblk}
\usepackage{fullpage}
\usepackage{times}
\usepackage[numbers]{natbib}
\usepackage{amsthm}
\usepackage{amsmath}
\usepackage{amssymb}
\usepackage{hyperref}
\usepackage[usenames,dvipsnames,svgnames,table]{xcolor}
\usepackage{framed}
\usepackage{graphicx}
\usepackage{tikz}
\usepackage{subfigure}
\usepackage{algorithm,algorithmic}
\usepackage[group-separator={,}]{siunitx}

\usetikzlibrary{shadows,arrows,shapes,positioning,calc}

\hypersetup{
    unicode=false,          
    pdftoolbar=true,        
    pdfmenubar=true,        
    pdffitwindow=false,     
    pdfstartview={FitH},    
    pdftitle={My title},    
    pdfauthor={Author},     
    pdfsubject={Subject},   
    pdfcreator={Creator},   
    pdfproducer={Producer}, 
    pdfkeywords={keyword1} {key2} {key3}, 
    pdfnewwindow=true,      
    colorlinks=true,       
    linkcolor=black,          
    citecolor=black,        
    filecolor=magenta,      
    urlcolor=cyan          
}

\newtheorem{theorem}{Theorem}[section]
\newtheorem{corollary}[theorem]{Corollary}

\newtheorem{lemma}[theorem]{Lemma}

\newtheorem{example}[theorem]{Example}
\newtheorem{claim}[theorem]{Claim}

\usepackage{enumitem}
\setenumerate{noitemsep,topsep=2pt,parsep=2pt,partopsep=2pt}

\newcommand{\lref}[2][]{\hyperref[#2]{#1~\ref*{#2}}}
\newcommand{\lrefeq}[2][]{\hyperref[#2]{#1~(\ref*{#2}})}

\newcommand{\expc}{\mathbb{E}}
\newcommand{\expcard}[1]{\expc[\card{#1}]}

\newcommand{\card}[1]{\lvert #1  \rvert}
\newcommand{\saize}{s_{\eta, k}}

\newcommand{\expset}[1]{\overline{K}(#1)}
\newcommand{\expmat}[1]{\overline{M}(#1)}
\newcommand{\condexpmat}[2]{\overline{M}(#1\vert #2)}
\newcommand{\condexpset}[2]{\overline{K}(#1\vert #2)}

\newcommand{\setapp}{\frac{2}{k} - \eta}

\mathchardef\mhyphen="2D

\newcommand{\msize}{M}
\newcommand{\ksize}[1]{K(#1)}
\renewcommand{\otimes}{\Delta}

\newcommand{\A}{{A}}
\newcommand{\B}{{B}}
\ifdefined\C
        \renewcommand{\C}{{C}}
    \else
        \newcommand{\C}{{C}}
    \fi

\renewcommand{\O}{{O}}

\newcommand{\red}[1]{{\color{black}#1}}
\newcommand{\blue}[1]{{\color{black}#1}}

\title{Ignorance is Almost Bliss:\\
 Near-Optimal Stochastic Matching With Few Queries \footnote{This work was supported by the NSF under grants CCF-1101668, CCF-1116892, CCF-1215883, CCF-1415460, IIS-0964579, IIS-1065251, IIS-1320620, and IIS-1350598, and by a National Defense Science \& Engineering Graduate Fellowship and a Sloan Research Fellowship. Authors' addresses: \texttt{\{avrim,dickerson,nhaghtal,arielpro,sandholm,ankits\}@cs.cmu.edu.}}
}

\author[1]{Avrim Blum}
\author[1]{John P. Dickerson} 
\author[1]{Nika Haghtalab}
\author[1]{Ariel D. Procaccia}
\author[1]{Tuomas Sandholm }
\author[2]{Ankit Sharma\footnote{This work was done while the author was a graduate student at Carnegie Mellon University.} }
\affil[1]{Carnegie Mellon University
}
\affil[2]{Solvvy Inc.
}

\begin{document}
\maketitle

\begin{abstract}
The stochastic matching problem deals with finding a maximum matching in a graph whose edges are unknown but can be accessed via queries.  This is a special case of stochastic $k$-set packing, where the problem is to find a maximum packing of sets, each of which exists with some probability.  In this paper, we provide edge and set query algorithms for these two problems, respectively, that provably achieve some fraction of the omniscient optimal solution.

Our main theoretical result for the stochastic matching (i.e., $2$-set packing) problem is the design of an \emph{adaptive} algorithm that queries only a constant number of edges per vertex and achieves a $(1-\epsilon)$ fraction of the omniscient optimal solution, for an arbitrarily small $\epsilon>0$. Moreover, this adaptive algorithm performs the queries in only a constant number of rounds. We complement this result with a \emph{non-adaptive} (i.e., one round of queries) algorithm that achieves a $(0.5 - \epsilon)$ fraction of the omniscient optimum.
We also extend both our results to stochastic $k$-set packing by designing an adaptive algorithm that achieves a $(\frac{2}{k} - \epsilon)$ fraction of the omniscient optimal solution, again with only $O(1)$ queries per element. This guarantee is close to the best known polynomial-time approximation ratio of $\frac{3}{k+1} -\epsilon$ for the \emph{deterministic} $k$-set packing problem \citep{Furer:2013}.

We empirically explore the application of (adaptations of) these algorithms to the kidney exchange problem, where patients with end-stage renal failure swap willing but incompatible donors.  We show on both generated data and on real data from the first \num{169} match runs of the UNOS nationwide kidney exchange that even a very small number of non-adaptive edge queries per vertex results in large gains in expected successful matches.

\end{abstract}



\maketitle


\section{Introduction}
\label{sec:intro}

In the \emph{stochastic matching} problem, we are given an undirected graph $G=(V,E)$, where we do not know which edges in $E$ actually exist. Rather, for each edge $e\in E$, we are given an existence probability $p_{e}$.  Of interest, then, are algorithms that first query some subset of edges to find the ones that exist, and based on these queries, produce a matching that is as large as possible.  The stochastic matching problem is a special case of \emph{stochastic $k$-set packing}, where each set exists only with some probability, and the problem is to find a packing of maximum size of those sets that do exist.

Without any constraints, one can simply query all edges or sets, and then output the maximum matching or packing over those that exist---but this level of freedom may not always be available.  We are interested in the tradeoff between the number of queries and the fraction of the omnsicient optimal solution achieved.
Specifically, we ask: In order to perform as well as the omniscient optimum in the stochastic matching problem, do we need to query (almost) all the edges, that is, do we need a budget of $\Theta(n)$ queries per vertex, where $n$ is the number of vertices? Or, can we, for any arbitrarily small $\epsilon>0$, achieve a $(1-\epsilon)$ fraction of the omniscient optimum by using an $o(n)$ per-vertex budget?  We answer these questions, as well as their extensions to the $k$-set packing problem.  We support our theoretical results empirically on both generated and real data from a large fielded kidney exchange in the United States.

\subsection{Our theoretical results and techniques}

Our main theoretical result gives a positive answer to the latter question for stochastic matching, by showing that, surprisingly, a \emph{constant} per-vertex budget is sufficient to get $\epsilon$-close to the omniscient optimum. Indeed, we design a polynomial-time algorithm with the following properties: for any constant $\epsilon>0$, the algorithm queries at most $O(1)$ edges incident to any particular vertex, requires $O(1)$ rounds of parallel queries, and achieves a $(1-\epsilon)$ fraction of the omniscient optimum.\footnote{This guarantee holds as long as all the non-zero $p_{e}$'s are bounded away from zero by some constant. The constant can be arbitrarily small but should not depend on $n$.} 

The foregoing algorithm is \emph{adaptive}, in the sense that its queries are conditioned on the answers to previous queries. Even though it requires only a constant number of rounds, it is natural to ask whether a non-adaptive algorithm---one that issues all its queries in one round---can also achieve a similar guarantee. We do not give a complete answer to this question, but we do present a non-adaptive algorithm that achieves a $0.5(1-\epsilon)$-approximation (for arbitrarily small $\epsilon>0$) to the omniscient optimum. We extend our matching results to a more general stochastic model in Appendix~\ref{sec:extensions}.

We extend our results to the stochastic $k$-set packing problem, where we are given a collection of sets, each with cardinality at most $k$. Stochastic Matching is a special case of Stochastic $k$-set packing: each set (which corresponds to an edge) has cardinality \num{2}, that is, $k=2$. In stochastic $k$-set packing, each set $s$ exists with some known probability $p_{s}$, and we need to query the sets to find whether they exist. Our objective is to output a collection of $\emph{disjoint}$ sets of maximum cardinality. 
We present adaptive and non-adaptive polynomial-time algorithms that achieve, for any constant $\epsilon>0$, at least $(\frac{2}{k}-\epsilon)$ and $(1-\epsilon)\frac{(2/k)^2}{2/k +1}$ fraction, respectively, of the omniscient optimum, again using $O(1)$ queries per element and hence $O(n)$ overall. For the sake of comparison, the best known \emph{polynomial-time} algorithm for optimizing $k$-set packing in the \emph{standard non-stochastic setting} has an approximation ratio of $\frac{3}{k+1} - \epsilon$~\citep{Furer:2013}. 

To better appreciate the challenge we face, we note that even in the stochastic matching setting, we do not have a clear idea of how large the omniscient optimum is. Indeed, there is a significant body of work on the expected cardinality of matching in \emph{complete} random graphs (see, e.g.,~\citep[Chapter 7]{Bol01}), where the omniscient optimum is known to be close to $n$. But in our work we are dealing with \emph{arbitrary} graphs where it can be a much smaller number. In addition, na\"ive algorithms fail to achieve our goal, even if they are allowed many queries. For example, querying a sublinear number of edges incident to each vertex, chosen uniformly at random, gives a vanishing fraction of the omniscient optimum---as we show in Appendix~\ref{app:examples}. 

The primary technical ingredient in the design of our \emph{adaptive algorithm} is that if, in any round $r$ of the algorithm, the solution computed by round $r$ (based on previous queries) is small compared to the omniscient optimum, then the current structure must admit a \emph{large collection of disjoint constant-sized `augmenting' structures}. These augmenting structures are composed of sets that have not been queried so far. Of course, we do not know whether these structures we are counting on to help augment our current matching actually exist; but we do know that these augmenting structures have constant size (and so each structure exists with some constant probability) and are \emph{disjoint} (and therefore the outcomes of the queries to the different augmenting structures are independent). Hence, by querying all these structures in parallel in round $r$, in expectation, we can close a constant fraction of the gap between our current solution and the omniscient optimum. By repeating this argument over a constant number of rounds, we achieve a $(1-\epsilon)$ fraction of the omniscient optimum. In the case of stochastic matching, these augmenting structures are simply augmenting paths; in the more general case of $k$-set packing, we borrow the notion of augmenting structures from \citet{Hurkens:1989}. 

\subsection{Our experimental results: Application to kidney exchange}
\label{sec:kidney}
Our work is directly motivated by applications to kidney exchange, a medical approach that enables kidney transplants. Transplanted kidneys are usually harvested from deceased donors; but as of April \num{26}, \num{2015}, there are \num{101671} people on the US national waiting list,\footnote{\url{http://optn.transplant.hrsa.gov}} making the median waiting time dangerously long. Fortunately, kidneys are an unusual organ in that donation by living donors is also a possibility---as long as patients happen to be medically compatible with their potential donors.

In its simplest form---\emph{pairwise exchange}---two incompatible donor-patient pairs exchange kidneys: the donor of the first pair donates to the patient of the second pair, and the donor of the second pair donates to the patient of the first pair.  This setting can be represented as an undirected \emph{compatibility graph}, where each vertex represents an incompatible donor-patient pair, and an edge between two vertices represents the possibility of a pairwise exchange. A matching in this graph specifies which exchanges take place. 

The edges of the compatibility graph can be determined based on the medical characteristics---blood type and tissue type---of donors and patients. However, the compatibility graph only tells part of the story. Before a transplant takes place, a more accurate medical test known as a \emph{crossmatch test} takes place. This test involves mixing samples of the blood of the patient and the donor (rather than simply looking up information in a database), making the test relatively costly and time consuming. Consequently, crossmatch tests are only performed for donors and patients that have been matched. While some patients are more likely to pass crossmatch tests than others---the probability is related to a measure of sensitization known as the person's Panel Reactive Antibody (PRA)---the average is as low as \num{30}\% in major kidney exchange programs~\citep{DPS13,Leishman13:Organ}. This means that, if we tested a perfect matching over $n$ donor-patient pairs, we would expect only $0.09n$ of the patients to actually receive a kidney. In contrast, the omniscient solution that runs crossmatch tests on all possible pairwise exchanges (in the compatibility graph) may be able to provide kidneys to all $n$ patients; but this solution is impractical.

Our adaptive algorithm for stochastic matching uncovers a sweet spot between these two extremes. On the one hand, it only mildly increases medical expenses, from one crossmatch test per patient, to a larger, yet constant, number; and it is highly parallelizable, requiring only a constant number of rounds, so the time required to complete all crossmatch tests does not scale with the number of donors and patients. On the other hand, the adaptive algorithm essentially recovers the entire benefit of testing all potentially feasible pairwise exchanges. 
The qualitative message of this theoretical result is clear: \emph{a mild increase in number of crossmatch tests provides nearly the full benefit of exhaustive testing.} 

The above discussion pertains to pairwise kidney exchange. However, modern kidney exchange programs regularly employ swaps involving three donor-patient pairs, which are known to provide significant benefit compared to pairwise swaps alone~\citep{RSU07,AR13}. Mathematically, we can consider a directed graph, where an edge $(u,v)$ means that the donor of pair $u$ is compatible with the patient of pair $v$ (before a crossmatch test was performed). In this graph, pairwise and \num{3}-way exchanges correspond to \num{2}-cycles and \num{3}-cycles, respectively. Our adaptive algorithm for \num{3}-set packing then provides a $(2/3)$-approximation to the omniscient optimum, using only $O(1)$ crossmatch tests per patient and $O(n)$ overall. While the practical implications of this result are currently not as crisp as those of its pairwise counterpart, future work may improve the approximation ratio (using $O(n)$ queries and an exponential-time algorithm), as we explain in Section~\ref{sec:conclusions-open}.

To bridge the gap between theory and practice, we provide experiments on both simulated data and real data from the first \num{169} match runs of the United Network for Organ Sharing (UNOS) US nationwide kidney exchange, which now includes \num{143} transplant centers---approximatey \num{60}\% of the transplant centers in the US.  The exchange began matching in October 2010 and now matches on a biweekly basis.  Using adaptations of the algorithms presented in this paper, we show that even a small number of non-adaptive rounds, followed by a single period during which only those edges selected during those rounds are queried, results in large gains relative to the omniscient matching.  We discuss the policy implications of this promising result in Section~\ref{sec:conclusions-policy}.

\section{Related work}
\label{sec:related-work}

While papers on stochastic matching often draw on kidney exchange for motivation---or at least mention it in passing---these two research areas are almost disjoint. We therefore discuss them separately in Sections~\ref{subsec:stochastic} and \ref{subsec:kidney}.

\subsection{Stochastic matching}
\label{subsec:stochastic}

Prior work has considered multiple variants of stochastic matching. A popular variant is the \emph{query-commit} problem, where the algorithm is \emph{forced} to add any queried edge to the matching if the edge is found to exist.
\blue{\citet{GT12} establish an upper bound of \num{0.7916} for graphs in which no information is available about the edges, while~\citet{CTT12} establish a lower bound of \num{0.573} and an upper bound of \num{0.898} for graphs in which each edge $e$ exists with a given probability $p_e$.} 
Similarly to our work, these approximation ratios are with respect to the omniscient optimum, but the informational disadvantage of the algorithm stems purely from the query-commit restriction. 

Within the query-commit setting, another thread of work~\citep{CIKM+09,Adam11,BGLM+12} imposes an additional \emph{per-vertex budget constraint} where the algorithm is not allowed to query more than a specified number, $b_{v}$, of edges incident to vertex $v$. With this additional constraint, the benchmark that the algorithm is compared to switches from the omniscient optimum to the constrained optimum, i.e., the performance of the best decision tree that obeys the per-vertex budget constraints and the query-commit restriction. In other words, the algorithm's disadvantage compared to the benchmark is only that it is constrained to run in polynomial time. Here, again, the best known approximation ratios are constant. A generalization of these results to packing problems has been studied by~\citet{GN13}.

Similarly to our work,~\citet{BGPS13} consider a stochastic matching setting without the query-commit constraint. They set the per-vertex budget to exactly $2$, and ask which subset of edges is queried by the optimal collection of queries subject to this constraint. They prove structural results about the optimal solution, which allow them to show that finding the optimal subset of edges to query is \textbf{NP}-hard. In addition, they give a polynomial-time algorithm that finds an almost optimal solution on a class of random graphs (inspired by kidney exchange settings). Crucially, the benchmark of~\citet{BGPS13} is also constrained to two queries per vertex. 

There is a significant body of work in stochastic optimization more broadly, for instance, the papers of \citet{Dean:2004} (Stochastic Knapsack), \citet{Gupta:2012} (Stochastic Orienteering), and \citet{Asadpur:2008} (Stochastic submodular maximization).

\subsection{Kidney exchange}
\label{subsec:kidney}
 
\red{Early models of kidney exchange did not explicitly consider the setting where an edge that is chosen to be matched only exists probabilistically.} 
Recent research by \citet{DPS13} and \citet{Anderson15:Finding} focuses on the kidney exchange application and restricts attention to a single crossmatch test per patient (the current practice), with a similar goal of maximizing the expected number of matched vertices, in a realistic setting (for example, they allow \num{3}-cycles and chains initiated by altruistic donors, who enter the exchange without a paired patient). They develop integer programming techniques, which are empirically evaluated using real and synthetic data.  \citet{Manlove15:Paired} discuss the integer programming formulation used by the national exchange in the United Kingdom, which takes edge failures into account in an ad hoc way by, for example, preferring shorter cycles to longer ones.  To our knowledge, our paper is the first to describe a general method for testing any number of edges \emph{before} the final match run is performed---and to provide experiments on real data showing the expected effect on fielded exchanges of such edge querying policies.
 
Another form of stochasticity present in fielded kidney exchanges is the arrival and departure of donor-patient pairs over time (and the associated arrival and departure of their involved edges in the compatibility graph).  Recent work has addressed this added form of dynamism from a theoretical~\cite{Unv10,Akbarpour14:Dynamic,Anderson15:Dynamic} and experimental~\cite{AS09,DPS12b,Dickerson15:FutureMatch} point of view.  Theoretical models have not addressed the case where an edge in the current graph may not exist (as we do in this paper); the more recent experimental papers have incorporated this possibility, but have not considered the problem of querying edges before recommending a final matching.  We leave as future research the analysis of edge querying in stochastic matching in such a dynamic model.

\section{The Model}
\label{sec:mod}
For any graph $G=(V,E)$, let $\msize(E)$ denote its maximum (cardinality) matching.\footnote{In the notation $\msize(E)$, we intentionally suppress the dependence on the vertex set $V$, since we care about the maximum matchings of different subsets of edges for a fixed vertex set.} In addition, for two matchings $M$ and $M'$, we denote their \emph{symmetric difference} by $M\otimes M'=(M\cup M')\setminus (M\cap M')$; it includes only paths and cycles consisting of alternating edges of $M$ and $M'$. 

In the stochastic setting, given a set of edges $X$, define $X_{p}$ to be the random subset formed by including each edge of $X$ independently with probability $p$. We will assume for ease of exposition that $p_e=p$ for all edges $e\in E$. Our results hold when $p$ is a lower bound, i.e., $p_e\geq p$ for all $e\in E$. Furthermore, in Appendix~\ref{sec:extensions}, we show that we can extend our results to a more general setting where the existence probabilities of edges incident to any particular vertex are correlated.

Given a graph $G=(V,E)$, define $\expmat{E}$ to be $\expc[\card{\msize(E_{p})}]$, where the expectation is taken over the random draw $E_{p}$. In addition, given the results of queries on some set of edges $T$, define $\condexpmat{E}{T}$ to be $\expc[\card{\msize(X_{p} \cup T')}]$, where $T'\subseteq T$ is the subset of edges of $T$ that are known to exist based on the queries, and $X = E \setminus T$.

In the \emph{non-adaptive} version of our problem, the goal is to design an algorithm that, given a graph $G=(V,E)$ with $|V|=n$, queries a subset $X$ of edges in parallel such that $\card{X}=O(n)$, and maximizes the ratio $\expmat{X}/\expmat{E}$. 

In contrast, an \emph{adaptive} algorithm proceeds in rounds, and in each round queries a subset of edges in parallel. Based on the results of the queries up to the current round, it can choose the subset of edges to test in the next round. Formally, an \emph{$R$-round adaptive} stochastic matching algorithm selects, in each round $r$, a subset of edges $X_{r}\subseteq E$, where $X_{r}$ can be a function of the results of the queries on \red{$\bigcup_{i < r}X_{i}$.}
The objective is to maximize the ratio $\expcard{\msize(\bigcup_{1\le i\le R}X_{i})}/\expmat{E}$, where the expectation in the numerator is taken over the outcome of the query results and the sets $X_{i}$ chosen by the algorithm.

\red{
To gain some intuition for our goal of arbitrarily good approximations to the omniscient optimum, and why it is challenging, let us consider a na\"ive algorithm and understand why it fails. This non-adaptive algorithm schedules $R = O(\log(n)/p)$ queries for each vertex as follows. First, order all vertices arbitrarily and start with an empty set of queries. In order, for each vertex $v$, let $N_R(v)$ be the set of neighbors of $v$ for whom at most $R$ queries have been scheduled. Schedule $\min\{R, N_R(v)\}$ queries, each between $v$ and an element of $N_R(v)$, where these elements are selected uniformly at random from $N_R(v)$. 

The next example shows that this proposed algorithm only achieves $\frac 56 $ fraction  of the omniscient optimal solution, as opposed to
our goal of achieving arbitrarily good  ($1-\epsilon$) approximations to the omniscient optimal. Furthermore, in the following example when each edge exists with probability $p>\frac 56$, this algorithm still  only achieves a
$\frac 56 $ fraction of the omniscient optimal solution, which is worse than a trivial algorithm of just picking \emph{one} maximum matching that guarantees a matching of size $pn$.

\begin{example}
Consider the graph $G=(V,E)$ whose vertices are partitioned into sets $A$, $B$, $C$, and $D$,  such that $|A| = |B| = \frac n2$ and $|C|= |D| = n$.
Let $E$  consist of two random bipartite graphs of degree $R = O(\log(n)/p)$  between $A$ and $B$ and similarly between $C$ and $D$. And let $B$ and $C$ be  connected with a \emph{complete} bipartite graph.  Let $p$ be the existence probability of any edge.

With high probability, there is a perfect matching that matches $A$ to $B$ and $C$ to $D$.
However, by the time the algorithm has processed half of the vertices, expected half of the vertices in $A$, $B$, $C$, and $D$ are processed.
For every vertex in $B$, this vertex has more neighbors in $C$ than in $D$. So, at this point, with high probability all vertices of $B$ already have $R$ queries scheduled from half of the vertices in $C$.
Therefore, after this point in the algorithm, no edges between $A$ and $B$ will be queried. So, $\frac12$ of the vertices in $A$ remain unmatched.
Compared to the omniscient optimum---which is a perfect matching with high probability---the approximation ratio of this algorithm is at most $\frac 56$.
\end{example}

For analysis of additional na\"ive algorithms refer to Appendix~\ref{app:examples}.

}

\section{Adaptive Algorithm: $(1-\epsilon)$-approximation}
\label{sec:ada-matching}

In this section, we present our main result: an adaptive algorithm---formally given as Algorithm~\ref{alg:adaptive-matching}---that achieves a $(1-\epsilon)$ approximation to the omniscient optimum for arbitrarily small $\epsilon>0$, using $O(1)$ queries per vertex and $O(1)$ rounds. 

The algorithm is initialized with the empty matching $M_0$. At the end of each round $r$, our goal is to maintain a maximum matching $M_r$ on the set of edges that are known to exist (based on queries made so far).
To this end, at round $r$, we compute the maximum matching $O_{r}$ on the set of edges that are known to exist  \emph{and} the ones that have not been queried yet (Step~\ref{algstep:optopt-max-mat}). We consider augmenting paths in $O_{r} \otimes M_{r-1}$, and query all the edges in them (Steps~\ref{step:ada-augmenting-paths} and \ref{algstep:ada-mat-query}). Based on the results of these queries $(Q_{r})$, we update the maximum matching $(M_{r})$. Finally, we return the maximum matching $M_{R}$ computed after $R=\frac{\log (2/ \epsilon)}{p^{2/\epsilon}}$ rounds. (Let us assume that $R$ is an integer for ease of exposition.)

\begin{algorithm}
\caption{\textsc{Adaptive Algorithm for Stochastic Matching: $(1-\epsilon)$ approximation }}
\label{alg:adaptive-matching}
{\bf Input}: A graph $G=(V, E)$.\\
{\bf Parameter}: $R = \frac{\log (2/ \epsilon)}{p^{2/\epsilon}}$
\begin{enumerate}
\item Initialize  $M_0$ to  the empty matching and $W_1\gets \emptyset$.
\item For $r = 1, \dots, R$, do \label{step:iter-matching}
\begin{enumerate}
\item \label{algstep:optopt-max-mat} Compute a maximum matching, $O_r$, in $(V, E\setminus W_r)$.
\item \label{step:ada-augmenting-paths} Set $Q_r$ to the collection of all augmenting paths \red{of $M_{r-1}$} in $O_r\otimes M_{r-1}$. 
\item \label{algstep:ada-mat-query} Query the edges in $Q_r$. Let $Q'_r$ and $Q''_r$ represent the set of existing and non-existing edges.
\item $W_{r+1} \gets W_r \cup Q''_r$. 
\item \label{algstep:ada-new-mat} Set  $M_r$ to the maximum matching in $\left(V, \bigcup_{j=1}^r Q'_j\right)$.
\end{enumerate}
\item Output $M_R$.
\end{enumerate}
\end{algorithm}

It is easy to see that \emph{this algorithm queries at most $\frac{\log (2/ \epsilon)}{p^{2/\epsilon}}$ edges per vertex}: In a given  round $r$, the algorithm queries edges that are in augmenting paths of $O_r\otimes M_{r-1}$. Since there is at most one augmenting path using any particular vertex, the algorithm queries at most one edge per vertex in each round.  Furthermore, the algorithm executes $\frac{\log (2/ \epsilon)}{p^{2/\epsilon}}$ rounds. Therefore, the number of queries issued by the algorithm per vertex is  as claimed.

The rest of the section is devoted to proving that the matching returned by this algorithm after $R$ rounds has cardinality that is, in expectation, at least a $(1-\epsilon)$ fraction of $\expmat{E}$. 

\begin{theorem}
\label{thm:main-adaptive-matching}
For any graph $G=(V, E)$ and any $\epsilon >0$, Algorithm~\ref{alg:adaptive-matching} returns a matching whose expected cardinality is at least $(1-\epsilon)~\expmat{E}$ in $R = \frac{\log (2/\epsilon)}{p^{(2/\epsilon)}}$ rounds. 
\end{theorem}

As mentioned in Section~\ref{sec:intro}, one of the insights behind this result is the existence of many \emph{disjoint} augmenting paths of \emph{bounded length} that can be used to augment a matching that is far from the omniscient optimum, that is, a lower bound on the number of elements in $Q_r$ of a given length $L$. This observation is formalized in the following lemma. (We emphasize that the lemma pertains to the non-stochastic setting.)

\begin{lemma}
\label{lem:struct-res-matching}
Consider a graph $G=(V,E)$ with two matchings $M_{1}$ and $M_{2}$. Suppose $\card{M_{2}} > \card{M_{1}}$. Then in $M_{1} \otimes M_{2}$, for any odd length $L \ge 1$, there exist at least $\card{M_{2}} - (1 + \frac{2}{L+1})\card{M_{1}}$ augmenting paths of length at most $L$, which augment the cardinality of $M_{1}$. \end{lemma}
\begin{proof}
Let $x_l$ be the number of augmenting paths of length $l$ (for any odd $l\ge 1$) found in $M_{1} \otimes M_{2}$ that augment the cardinality of $M_{1}$. Each augmenting path increases the size of $M_{1}$ by $1$, so the total number of augmenting paths $\sum_{l\geq 1} x_{l}$ is at least $|M_{2}| - |M_{1}|$. Moreover, each augmenting path of length $l$ has $\frac{l-1}{2}$ edges in $M_{1}$. Hence, $\sum_{l\geq 1} \frac{l-1}{2}\,x_{l} \leq \card{M_{1}}$. In particular, this implies that $\frac{L+1}{2}\sum_{l \geq L+2} x_l \leq \card{M_{1}}$. We conclude that
\begin{align*}
\sum_{l = 1}^{L} x_{l} &=  \sum_{l\geq 1} x_l - \sum_{l\geq L+2} x_l  \geq \left(\card{M_{2}}- \card{M_{1}}\right) -  \frac{2}{L+1}   \card{M_{1}}  =   \card{M_{2}}-\left(1+\frac{2}{L+1}\right) \card{M_{1}}.
\end{align*}
\end{proof}

The rest of the theorem's proof requires some additional notation.
 At the beginning of any given round $r$, the algorithm already knows about the existence (or non-existence) of the edges in $\bigcup_{i = 1}^{r-1}Q_{i}$.
 We use $Z_{r}$ to denote the expected size of the maximum matching in graph $G=(V,E)$ given the results of the queries $\bigcup_{i = 1}^{r-1}Q_{i}$. More formally, $Z_{r} = \condexpmat{E}{\bigcup_{i = 1}^{r-1}Q_{i}}$. Note that $Z_{1} = \expmat{E}$.

For a given $r$, we use the notation $\expc_{Q_r}[X]$ to denote the expected value of $X$ where the expectation is taken \emph{only} over the outcome of query $Q_r$, and fixing the outcomes on the results of queries $\bigcup_{i=1}^{r-1}Q_i$. Moreover, for a given $r$, we use $\expc_{Q_r, \dots, Q_R}[X]$ to denote the expected value of $X$ with the expectation taken over the outcomes of queries $\bigcup_{i=r}^{R}Q_i$, and fixing an outcome on the results of queries $\bigcup_{i=1}^{r-1}Q_i$.

In Lemma~\ref{lem:increase-each-iter-matching}, for any round $r$ and for \emph{any} outcome of the queries $\bigcup_{i = 1}^{r-1}Q_{i}$, we lower-bound the \emph{expected increase in the size of $M_{r}$} over the size of $M_{r-1}$, with the expectation being taken only over the outcome of edges in $Q_{r}$. This lower bound is a function of $Z_{r}$.

\begin{lemma}
\label{lem:increase-each-iter-matching}
For any $r\in [R]$, odd $L$, and $Q_{1}, \cdots, Q_{r-1}$, it holds that $\expc_{Q_{r}}[|M_r|] \geq (1- \gamma) |M_{r-1}| + \alpha Z_{r}$, where $\gamma = p^{(L + 1)/2}~\left(1+ \frac{2}{L+1}\right)$ and $\alpha = p^{(L+1)/2}$. 
\end{lemma}
\begin{proof}
By Lemma~\ref{lem:struct-res-matching}, there exist at least $\card{O_{r}} - (1+\frac{2}{L+1})\card{M_{r-1}}$ augmenting paths in $O_{r} \otimes M_{r-1}$ that augment $M_{r-1}$ and are of length at most $L$. The $O_r$ part of every augmenting path of length at most $L$ exists independently with probability at least $p^{(L+1)/2}$. Therefore, the expected increase in the size of the matching is:
\begin{align*}
\expc_{Q_r}[|M_{r}|] - |M_{r-1}| & ~\geq~    p^{\frac{L+1}{2}}   \left( |O_r|-\left(1+\frac{2}{L+1}\right)|M_{r-1}| \right) \\ 
& ~=~ \alpha |O_r| - \gamma |M_{r-1}|  ~\geq~  \alpha Z_r - \gamma |M_{r-1}|,
\end{align*}
where the last inequality holds by the fact that $Z_r$, which is the expected size of the optimal matching with expectation taken over non-queried edges, cannot be larger than $O_r$, which is  the maximum matching assuming that every non-queried edge exists. 
\end{proof}

We are now ready to prove the theorem.

\begin{proof}[\textsc{of Theorem~\ref{thm:main-adaptive-matching}}]
Let $L= \frac{4}{\epsilon}-1$; it is assumed to be an odd integer for ease of exposition.\footnote{Otherwise there exists $\epsilon/2\leq \epsilon'\leq \epsilon$ such that $\frac{4}{\epsilon'}-1$ is an odd integer. We use a similar simplification in the proofs of other results in the appendix.}  
By Lemma~\ref{lem:increase-each-iter-matching}, we know that for every $r\in [R]$, $\expc_{Q_{r}}[\card{M_r|} \geq (1- \gamma) \card{M_{r-1}} + \alpha Z_{r}$, where $\gamma  = p^{(L+1)/2}(1 + \frac{2}{L+1})$, and $\alpha = p^{(L+1)/2}$. We will use this inequality repeatedly to derive our result. We will also require the equality
\begin{equation}
\label{eq:jeefa}
\expc_{Q_{r-1}}[Z_{r}] = \expc_{Q_{r-1}} \left[  \condexpmat{E}{\bigcup_{i = 1}^{r-1}Q_{i}} \right]  = \condexpmat{E}{\bigcup_{i = 1}^{r-2}Q_{i}}  = Z_{r-1}.
\end{equation}

First, applying Lemma~\ref{lem:increase-each-iter-matching} at round $R$, we have that $\expc_{Q_{R}}[\card{M_{R}}] \geq (1- \gamma) \card{M_{R-1}} + \alpha Z_{R}$. This inequality is true for any fixed outcomes of $Q_1,\ldots,Q_{R-1}$. In particular, we can take the expectation over $Q_{R-1}$, and obtain $$\expc_{Q_{R-1}, Q_{R}}[\card{M_{R}}] \geq (1- \gamma)~\expc_{Q_{R-1}}[\card{M_{R-1}}] + \alpha~\expc_{Q_{R-1}}[Z_{R}].$$
By Equation~\eqref{eq:jeefa}, we know that $\expc_{Q_{R-1}}[Z_{R}] = Z_{R-1}$. Furthermore, we can apply Lemma~\ref{lem:increase-each-iter-matching} to $\expc_{Q_{R-1}}[\card{M_{R-1}}]$ to get the following inequality:
\begin{align*}
\expc_{Q_{R-1}, Q_{R}}[\card{M_{R}}] &\geq (1- \gamma)~\expc_{Q_{R-1}}[\card{M_{R-1}}] + \alpha ~\expc_{Q_{R-1}}[Z_{R}] \\
&\geq (1- \gamma)~\left((1- \gamma) ~\card{M_{R-2}} + \alpha~Z_{R-1}\right) + \alpha ~Z_{R-1}\\
&= (1-\gamma)^{2}~\card{M_{R-2}} + \alpha~(1 + (1-\gamma))~Z_{R-1}.
\end{align*}

We repeat the above steps by sequentially taking expectations over $Q_{R-2}$ through $Q_{1}$, and at each step applying Equation~\eqref{eq:jeefa} and Lemma~\ref{lem:increase-each-iter-matching}. This gives us 
\begin{align*}
\expc_{Q_1, \dots, Q_{R}} [\card{M_R}] & ~\geq~ (1-\gamma)^{R}\card{M_{0}} + \alpha~(1+(1-\gamma) + \cdots + (1-\gamma)^{R-1}) Z_{1} \\
& ~=~ \alpha~\frac{1 - (1-\gamma)^{R}}{\gamma}Z_1,
\end{align*}
where the second transition follows from the initialization of $M_{0}$ as an empty matching.
Since $L= \frac{4}{\epsilon}-1$ and $R = \frac{\log (2/ \epsilon)}{p^{2/\epsilon}}$, we have 
\begin{align}\label{eq:matching-epsilon}
\frac{\alpha}{\gamma} \big( 1-(1-\gamma)^R \big) = \left(1 - \frac{2}{L+ 1}\right) \big( 1-(1-\gamma)^R \big) \geq 1 - \frac{2}{L + 1} - e^{-\gamma R} \geq 1 - \frac{\epsilon}{2} - \frac{\epsilon}{2}  = 1- \epsilon,
\end{align}
where the second transition is true because $e^{-x}\geq 1-x$ for all $x\in \mathbb{R}$. 
We conclude that $\expc_{Q_1, \dots, Q_{R}} [\card{M_R}] \ge (1-\epsilon)~Z_1$.
Because $Z_1 = \expmat{E}$, it follows that the expected size of the algorithm's output is at least  $(1-\epsilon)~\expmat{E}$. 
\end{proof}
 
In Appendix~\ref{sec:extensions}, we extend our results to the setting where edges have correlated existence probabilities---an edge's probability is determined by parameters associated with its two vertices. This generalization gives a better model for kidney exchange, as some patients are \emph{highly sensitized} and therefore harder to match in general; this means that all edges incident to such vertices are less likely to exist.
We consider two settings, first, where an adversary chooses the vertex parameters, and second, where these parameters are drawn from a distribution. 
Our approach involves excluding from our analysis edges whose existence probability is too low.
We do so by showing that (under specific conditions) excluding any augmenting path that includes such edges still leaves us with a large number of constant-size augmenting paths.

\section{Non-adaptive algorithm: $0.5$-approximation}
\label{sec:non-adaptive-matching} \label{SEC:NON-ADAPTIVE-MATCHING}

The adaptive algorithm, Algorithm~\ref{alg:adaptive-matching}, augments the current matching by computing a maximum matching on queried edges that are known to exist, and edges that have not been queried. One way to extend this idea to the non-adaptive setting is the following: we can simply choose several edge-disjoint matchings, and hope that they help in augmenting each other. In this section, we ask: How close can this non-adaptive interpretation of our adaptive approach take us to the omniscient optimum?

In more detail, our non-adaptive algorithm---formally given as Algorithm~\ref{alg:non-adaptive-matching}---iterates $R=\frac{\log (2/ \epsilon)}{p^{2/\epsilon}}$ times. In each iteration, it picks a maximum matching and removes it. The set of edges queried by the algorithm is the union of the edges chosen in some iteration. We will show that, for any arbitrarily small $\epsilon>0$, the algorithm finds a $0.5(1-\epsilon)$-approximate solution.  Since we allow an arbitrarily small (though constant)
probability $p$ for stochastic matching, achieving a $0.5$-approximation independently of the value of $p$, while querying only a linear number of edges, is nontrivial. For example, a na\"ive algorithm that only queries one maximum matching clearly does not guarantee a $0.5$-approximation---it would guarantee only a $p$-approximation. In addition, the example given in Section~\ref{sec:mod} shows that choosing edges at random performs poorly.

\begin{algorithm}
\caption{\textsc{Non-adaptive algorithm for Stochastic Matching: $0.5$-approximation}}
\label{alg:non-adaptive-matching}
{\bf Input}: A graph $G(V, E)$.\\
{\bf Parameter}: $R = \frac{\log (2/ \epsilon)}{p^{2/\epsilon}}$
\begin{enumerate}
\item Initialize $W_1\gets \emptyset$.
\item For $r = 1, \dots, R$, do \label{step:non-ada-iter-matching}
\begin{enumerate}
\item Compute a maximum matching, $O_{r}$, in $\left(V, E\setminus \bigcup_{1\le i \le r-1}W_{i}\right)$.
\item \label{step:non-ada-augmenting-paths} $W_{r} \leftarrow W_{r-1} \cup O_{r}$.
\end{enumerate}
\item Query all the edges in $W_{R}$, and output the maximum matching among the edges that are found to exist in $W_{R}$.
\end{enumerate}
\end{algorithm}

The number of edges incident to any particular vertex that are queried by the algorithm is at most $\frac{\log (2/ \epsilon)}{p^{2/\epsilon}}$, because the vertex can be matched with at most one neighbor in each round. The next theorem, whose proof appears in Appendix~\ref{app:non-adap}, establishes the approximation guarantee of Algorithm~\ref{alg:non-adaptive-matching}.

\begin{theorem}
\label{thm:non-adaptive-alg}
Given a graph $G=(V,E)$ and any $\epsilon >0$, the expected size $\expmat{W_{R}}$ of the matching produced by Algorithm~\ref{alg:non-adaptive-matching} is at least a $0.5 (1- \epsilon)$ fraction of $\expmat{E}$.
\end{theorem}

As explained in Section~\ref{sec:conclusions-open}, we do not know whether in general non-adaptive algorithms can achieve a $(1-\epsilon)$-approximation with $O(1)$ queries per vertex. However, if there is such an algorithm, it is not Algorithm~\ref{alg:non-adaptive-matching}! Indeed, the next theorem (whose proof is relegated to Appendix~\ref{app:non-adap}) shows that the algorithm cannot give an approximation ratio better than $5/6$ to the omniscient optimum. 
This fact holds even when $R=\Theta(\log n)$. 

\begin{theorem}\label{thm:non-adaptive-upperbound}
Let $p=0.5$. For any $\epsilon>0$ there exists $n$ and a graph $(V, E)$ with $|V|\geq n$ such that Algorithm~\ref{alg:non-adaptive-matching}, with $R=O(\log n)$, returns a matching with expected size of at most $\frac 56 \expmat{E} + \epsilon$. 
\end{theorem}

Despite this somewhat negative result, in Section~\ref{sec:experiments}, we show experimentally on realistic kidney exchange compatibility graphs that Algorithm~\ref{alg:non-adaptive-matching} performs very well for even very small values of $R$, across a wide range of values of $p$.

\section{Generalization to $k$-Set Packing}
\label{sec:kset} \label{SEC:KSET}

So far we have focused on stochastic matching, for ease of exposition. But our approach directly generalizes to the $k$-set packing problem. Here we describe this extension for the adaptive (Section~\ref{sec:kset-adaptive-main}) and non-adaptive (Section~\ref{sec:kset-non-adaptive-main}) cases, and relegate the details---in particular, most proofs---to Appendix~\ref{app:kset}. 

Formally, a $k$-set packing instance $(U,A)$ consists of a set of elements $U$, $|U|=n$, and a collection of subsets $A$, such that each subset $S$ in $A$ contains at most $k$ elements of $U$, that is, $S \subseteq U$ and $\card{S} \le k$. Given such an instance,  a feasible solution is a collection of sets $B \subseteq A$ such that any two sets in $B$ are disjoint. We use $\ksize{A}$ to denote the largest feasible solution $B$. 

Finding an optimal solution to the $k$-set packing problem is \textbf{NP}-hard (see, e.g., \citep{ABS07} for the special case of $k$-cycle packing). \citet{Hurkens:1989} designed a polynomial-time local search algorithm with an approximation ratio of $(\frac{2}{k} - \eta)$, using local improvements of \emph{constant size} that depends only on $\eta$ and $k$. We denote this constant by $\saize$.  More formally, consider an instance $(U,A)$ of $k$-set packing and let $B \subseteq A$ be a collection of disjoint $k$-sets. $(C,D)$ is said to be an \emph{augmenting structure} for $B$ if removing $D$ and adding $C$ to $B$ increases the cardinality and maintains the disjointness of the resulting collection, i.e., if $(B \cup C) \setminus D$ is a disjoint collection of $k$-sets and $\card{(B \cup C) \setminus D} > \card{B}$, where $C \subseteq A$ and $D \subseteq B$. 

\citet{Hurkens:1989} have also shown that an approximation ratio better than $2/k$ cannot be achieved with structures of constant size. While subsequent work~\citep{Furer:2013} has improved the approximation ratio, their local search algorithm finds structures of super-constant size. This is inconsistent with our technical approach, as we need each queried structure to exist (in its entirety) with constant probability.

To be more precise, \citet{Hurkens:1989} prove: 
\begin{lemma}[\citep{Hurkens:1989}]
\label{lem:hurkens}
Given a collection $B$ of disjoint sets such that $\card{B} < (2/k-\eta)\card{\ksize{A}}$, there exists an augmenting structure $(C,D)$ for $B$ such that both $C$ and $D$ have at most $\saize$ sets, for a constant $\saize$ that depends only on $\eta$ and $k$.
\end{lemma}

However, we need to find \emph{many} augmenting structures. We use Lemma~\ref{lem:hurkens} to prove:

\begin{lemma}
\label{lem:struct-res-k-set}
If $\card{B} < \card{\ksize{A}}$, then there exist $\frac{1}{k\,\saize}(\card{\ksize{A}} - \frac{\card{B}}{\frac{2}{k} - \eta})$ disjoint augmenting structures that augment the cardinality of $B$, each with size at most $\saize$. Moreover, this collection of augmenting structures can be found in polynomial time.
\end{lemma}

\begin{proof} 
We prove the lemma using  Algorithm~\ref{alg:k-set-aug}. We claim that if we run this algorithm on the $k$-set packing instance $(U,A)$ and the collection $B$, then it will return a collection $Q$ of at least $T=\frac{1}{k\,\saize}(\card{\ksize{A}} - \frac{\card{B}}{\setapp})$ disjoint augmenting structures $(C,D)$ for $B$. By Step~\ref{algstep:removal-of-C}, we are guaranteed that $Q$ consists of \emph{disjoint} augmenting structures. Hence, all that is left to show is that in each of the first $T$ iterations, at Step~\ref{algstep:aug-structure}, we are able to find a \emph{nonempty} augmenting structure $(C,D)$ for $B$.

By Lemma~\ref{lem:hurkens}, we know that if at iteration $t$ it is the case that $\card{B} < (\setapp) \card{\ksize{A_{t}}}$, then we will find an augmenting structure $(C,D)$ of size $\saize$ for $B$. To prove that the inequality holds at each iteration $t\le T$, we first claim that for all $t$, 
\begin{equation}
\label{eq:jeefa-k}
\card{\ksize{A_{t}}} \ge \card{\ksize{A}} - (t-1)\cdot k\cdot \saize
\end{equation}
We prove this by induction. The claim is clearly true for the base case of $t=1$. For the inductive step, suppose it is true for $t$, then we know that $\card{\ksize{A_{t}}} \ge \card{\ksize{A}} - (t-1)\cdot k\cdot \saize$.
At iteration $t$, the augmenting structure $(C,D)$ can intersect with at most $\saize\cdot k$ sets of $\ksize{A_{t}}$. This is true since $\ksize{A_{t}}$ consists of \emph{disjoint} sets, and the augmenting structure $(C,D)$ is of size at most $\saize$.
Hence, Step~\ref{algstep:removal-of-C}  reduces $\card{\ksize{A_t}}$  by at most  $k\cdot \saize$. So,
$\card{\ksize{A_{t+1}}} \geq \card{\ksize{A_{t}}} - k\cdot \saize$. Combining the two inequalities, $\card{\ksize{A_{t}}} \ge \card{\ksize{A}} - (t-1)\cdot k\cdot \saize$ and $\card{\ksize{A_{t+1}}} \ge \card{\ksize{A_{t}}} - k\cdot \saize$, we have $\card{\ksize{A_{t+1}}} \ge \card{\ksize{A}} - t\cdot k\cdot \saize$. This establishes Equation~\eqref{eq:jeefa-k}.

We conclude that if the for-loop adds \emph{non-empty} augmenting structures only for the first $t$ rounds, it must be the case that $\card{B} \ge (\setapp) \card{\ksize{A_{t+1}}}$, and therefore $\card{B} \ge (\setapp)(\card{\ksize{A}} - t\cdot k\cdot \saize)$ which implies that $t \ge \frac{1}{k\,\saize}( \card{\ksize{A}} - \frac{\card{B}}{\setapp})$.
\end{proof}

\begin{algorithm}[ht]
\caption{\textsc{Finding constant-size disjoint Augmenting structures for $k$-sets}}
\label{alg:k-set-aug}
{\bf Input}: $k$-set packing instance $(U,A)$ and a collection $B\subseteq A$ of disjoint sets.\\
{\bf Output}: Collection $Q$ of disjoint augmenting structures $(C,D)$ for $B$.\\
{\bf Parameter}: $\saize$ (the desired maximum size of the augmenting structures)
\begin{enumerate}
\item Initialize $A_{1} \leftarrow A$ and $Q \leftarrow \phi$ (empty set).
\item For $t=1,\cdots, \card{A}$
\begin{enumerate}
\item\label{algstep:aug-structure} Find an augmenting structure $(C,D)$ of size $\saize$ for $B$ on the $k$-set instance $(U, A_{t})$.
\item Add $(C,D)$ to $Q$. (If $C$ is an empty set, break out of the loop.)
\item\label{algstep:removal-of-C} Set $A_{t+1}$ to be  $A_{t}$ minus the collection $C$ and any set in $A_{t} \setminus B$ that intersects with $C$.
\end{enumerate}
\item Output $Q$.
\end{enumerate}
\end{algorithm}

\subsection{Adaptive algorithm for $k$-set packing}
\label{sec:kset-adaptive-main}
Turning to the stochastic version of the problem, given $(U,A)$, let $A_{p}$ be a random subset of $A$ where each set from $A$ is included in $A_{p}$ independently with probability $p$. We then define $\expset{A}$ to be $\expc[\card{\ksize{A_{p}}}]$, where the expectation is taken over the random draw $A_{p}$. Similarly to the matching setting, this is the omniscient optimum---our benchmark.

We extend the ideas introduced earlier in the paper for matching, together with Lemma~\ref{lem:struct-res-k-set} and additional ingredients, to obtain the following result for the adaptive problem.

\begin{theorem}
\label{thm:main-adaptive-kset}
There exists an adaptive polynomial-time algorithm that, given a $k$-set instance $(U,A)$ and $\epsilon>0$, uses $O(1)$ rounds and $O(n)$ queries overall, and returns a set $B_{R}$ whose expected cardinality is at least a $(1-\epsilon)\frac{2}{k}$ fraction of $\expset{A}$.
\end{theorem}

With an eye toward Theorem~\ref{thm:main-adaptive-kset}, Algorithm~\ref{alg:adaptive-kset} is a polynomial-time algorithm that can be used to find such a packing that approximates the omniscient optimum.  In each round $r$, the algorithm maintains a feasible $k$-set packing $B_{r}$ based on the $k$-sets that have been queried so far. It then computes a collection $Q_{r}$ of disjoint, small augmenting structures with respect to the current solution $B_{r}$, where the augmenting structures are composed of sets that have not been queried so far. It issues queries to these augmenting structures, and uses those that are found to exist to augment the current  solution. The augmented solution is fed into the next round.

\begin{algorithm}
\caption{\textsc{Adaptive Algorithm for Stochastic $k$-Set packing}}
\label{alg:adaptive-kset}
{\bf Input}: A $k$-set instance $(U,A)$, and $\epsilon>0$.\\
{\bf Parameters}: $\eta = \frac \epsilon k$ and $R=\frac {(\setapp)~k~\saize}{p^{\saize}} \log( \frac 2\epsilon)~~~$ (For a $(1-\epsilon)(\frac 2k)$-approximation) 
\begin{enumerate}
\item Initialize $r\gets 1$, $B_{0} \gets \emptyset$ and $A_{1} \gets A$.
\item For $r = 1, \dots, R$, do \label{step:iter}
\begin{enumerate}
\item Initialize $B_{r}$ to $B_{r-1}$.
\item Let $Q_{r}$ be the set of augmenting structures 
given by Algorithm~\ref{alg:k-set-aug} 
on the input consisting of the $k$-set packing instance $(U, A_{r})$, the collection $B_{r}$, and the parameter $\saize$.
\item\label{step:increment-B} For each augmenting structure $(C,D) \in Q_r$.
\begin{enumerate}
\item Query all sets in $C$.
\item If all the sets of $C$ exist, augment the current solution: $B_{r} \gets (B_{r} \setminus D) \cup C$.
\end{enumerate}
\item Set $A_{r+1}$ to be $A_r$ after removing queried sets that were found not to exist.
\end{enumerate}
\item Return $B_{R}$.
\end{enumerate}
\end{algorithm}

Similarly to our matching results, for any element $v \in U$, the number of sets that it belongs to and are also queried is at most $R$. Indeed, in each of the $R$ rounds, Algorithm~\ref{alg:adaptive-kset} issues queries to disjoint augmenting structures, and each augmenting structure includes at most one set per element.

\subsection{Non-adaptive algorithm for $k$-set packing}
\label{sec:kset-non-adaptive-main}

Once again, when going from the adaptive case to the non-adaptive case, the fraction of the omniscient optimum that we can obtain becomes significantly worse. 

\begin{theorem} \label{thm:non-adaptive-sets}
There exists a non-adaptive polynomial-time algorithm that, given a $k$-set instance $(U, \A)$ and $\epsilon>0$, uses $O(n)$ queries overall and returns a $k$-set packing with expected cardinality  $(1-\epsilon)\frac{(2/k)^2}{2/k +1} \expset{\A}$.
\end{theorem}

We present such a polynomial-time non-adaptive algorithm, Algorithm \ref{alg:non-adaptive-sets}, that proceeds as follows.
For $R$ rounds, at every round, using the local improvement algorithm of \citet{Hurkens:1989}, we find a $(\setapp)$-approximate solution to the $k$-set instance and remove it. Then, we query every set that is included in these $R$ solutions. We show that the expected cardinality of the maximum packing on the chosen sets is a $(1-\epsilon ) \frac{(2/k)^2}{2/k + 1}$ approximation of the expected optimal packing. As usual, it is easy to see that $O(n)$ queries are issued overall.

Importantly, the statements of Theorems~\ref{thm:main-adaptive-matching} and \ref{thm:non-adaptive-alg} are special cases of the statements of Theorems \ref{thm:main-adaptive-kset} and \ref{thm:non-adaptive-sets}, respectively, for $k=2$, although on a technical level the $k=2$ case must be handled separately (as we did, with less cumbersome terminology and technical constructions than in the general $k$-set packing setting).

\begin{algorithm}
\caption{\textsc{Non-Adaptive Algorithm for Stochastic $k$-Set packing}}
\label{alg:non-adaptive-sets}
{\bf Input}: A $k$-set packing instance $(U,A)$, and $\epsilon>0$.\\
{\bf Parameters:} $\eta  = \frac{\epsilon}{2k}$ and $R=\frac {(\setapp)~k~\saize}{p^{\saize}} \log( \frac 2\epsilon).~~~$ (For $(1-\epsilon)\frac{(2/k)^2}{2/k + 1}$-approximation)
\begin{enumerate}
\item Let $\B_{0} \gets \emptyset$.
\item For $r = 1, \dots, R$, do \label{step:iter}
\begin{enumerate}
\item $O_{r}\gets$ a $(\setapp)$-approximate solution to the $k$-set instance $(U,\A\setminus \bigcup_{i=1}^{r-1} \B_i)$. ($O_{r}$ is found using the local improvement algorithm of \citet{Hurkens:1989}.) 
\item Set $\B_r \gets \B_{r-1} \cup \O_r$.
\end{enumerate}
\item Query the sets in $O_{1}$, and assign $Q_{1}$ to be the sets that are found to exist. \label{step:non-ada-kset-stage1}
\item For $r=2, \cdots, R$, do \label{step:non-ada-kset-aug}
\begin{enumerate}
\item Find augmenting structures in $O_{r}$ that augment $Q_{r-1}$. This is achieved by giving the instance $(U,Q_{r-1} \cup O_{r})$ and solution $Q_{r-1}$ as input (with parameter $\saize$) to Algorithm~\ref{alg:k-set-aug}.
\item Query all the augmenting structures in $O_{r}$, and augment $Q_{r-1}$ with the ones that are found to exist. Call the augmented solution $Q_{r}$.
\end{enumerate}
\item Output $Q_{R}$.
\end{enumerate}
\end{algorithm}

\section{Experimental results on kidney exchange compatibility graphs}
\label{sec:experiments}

In this section, we support our theoretical results with empirical simulations from two kidney exchange compatibility graph distributions.  The first distribution, due to~\citet{SRSU+06}, was designed to mimic the characteristics of a nationwide exchange in the United States in steady state.  Fielded kidney exchanges have not yet reached that point, though; with this in mind, we also include results on \emph{real} kidney exchange compatibility graphs drawn from the first \num{169} match runs of the UNOS nationwide kidney exchange.  While these two families of graphs differ substantially, we find that even a small number $R$ of non-adaptive rounds, followed by a single period during which only those edges selected during the $R$ rounds are queried, results in large gains relative to the omniscient matching.

As is common in the kidney exchange literature, in the rest of this section we will loosely use the term ``matching'' to refer to both $2$-set packing (equivalent to the traditional definition of matching, where two vertices connected by directed edges are translated to two vertices connected by a single undirected edge) and $k$-set packing, possibly with the inclusion of altruist-initiated chains.

This section does not directly test the algorithms presented in this paper.  For the \num{2}-cycles-only case, we do directly implement Algorithm~\ref{alg:non-adaptive-matching}.  However, for the cases involving longer cycles and/or chains, we do not restrict ourselves to polynomial time algorithms (unlike in the theory part of this paper), instead choosing to optimally solve matching problems using integer programming during each round, as well as for the final matching and for the omniscient benchmark matching.  This decision is informed by the current practice in kidney exchange, where computational resources are much less of a problem than human or monetary resources (of which the latter two are necessary for querying edges).  

In our experiments, the planning of which edges to query proceeds in rounds as follows.  Each round of matching calls as a subsolver the matching algorithm presented by~\citet{DPS13}, which includes edge failure probabilities in the optimization objective to provide a maximum discounted utility matching.  The set of cycles and chains present in a round's discounted matching are added to a set of edges to query, and then those cycles and chains are constrained from appearing in future rounds.  After all rounds are completed, this set of edges is queried, and a final maximum discounted utility matching is compared against an omniscient matching that knows the set of non-failing edges up front.

\subsection{Experiments on dense generated graphs due to~\citet{SRSU+06}}
\label{sec:experiments-saidman}

We begin by looking at graphs drawn from a distribution due to~\citet{SRSU+06}, hereafter referred to as ``the Saidman generator.''  This generator takes into account the blood types of patients and donors (such that the distribution is drawn from the general United States population), as well as three levels of PRA and various other medical characteristics of patients and donors that may affect the existence of an edge.  Fielded kidney exchanges currently do not uniformly sample their pairs from the set of all needy patients and able donors in the US, as assumed by the Saidman generator; rather, exchanges tend to get hard-to-match patients who have not received an organ through other means.  Because of this, the Saidman generator tends to produce compatibility graphs that are significantly denser than those seen in fielded kidney exchanges today (see, e.g.,~\cite{AGRR11,AJM13}).

Figure~\ref{fig:experiments-saidman} presents the fraction of the omniscient objective achieved by $R \in \{0,1,\ldots,5\}$ non-adaptive rounds of edge testing for generated graphs with \num{250} patient-donor pairs and no altruistic donors, constrained to \num{2}-cycles only (left) and both \num{2}- and \num{3}-cycles (right).  Note that the case $R=0$ corresponds to no edge testing, where a maximum discounted utility matching is determined by the optimizer and then compared directly to the omniscient matching.  The x-axis varies the uniform edge failure rate $f$ from \num{0.0}, where edges do not fail, to \num{0.9}, where edges only succeed with a \num{10}\% probability.  Given an edge failure rate of $f$ in the figures below, we can translate to the $p$ used in the theoretical section of the paper as follows: \num{2}-cycles exists with probability $p_{\text{\num{2}-cycle}} = (1-f)^2$, while a \num{3}-cycle exists with $p_{\text{\num{3}-cycle}} = (1-f)^3$.  For example, in the case of $f = 0.9$, a \num{3}-cycle exists with very low probability $p = 0.001$.

\begin{figure}[ht!bp]
\centering
\begin{minipage}[b]{0.48\textwidth}
\centering
\includegraphics[width=\textwidth]{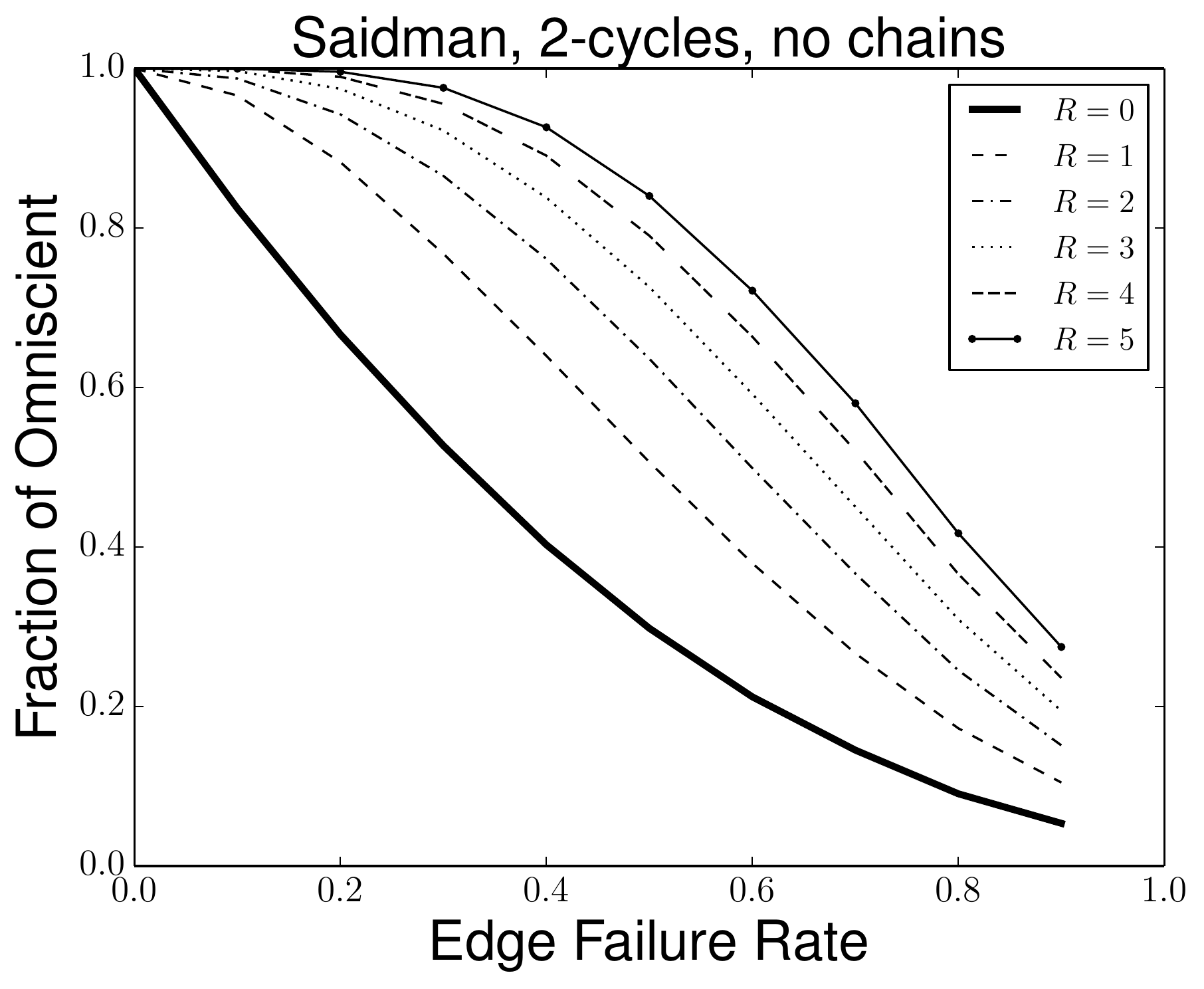}
\end{minipage}%
\hfill%
\begin{minipage}[b]{0.48\textwidth}
\centering
\includegraphics[width=\textwidth]{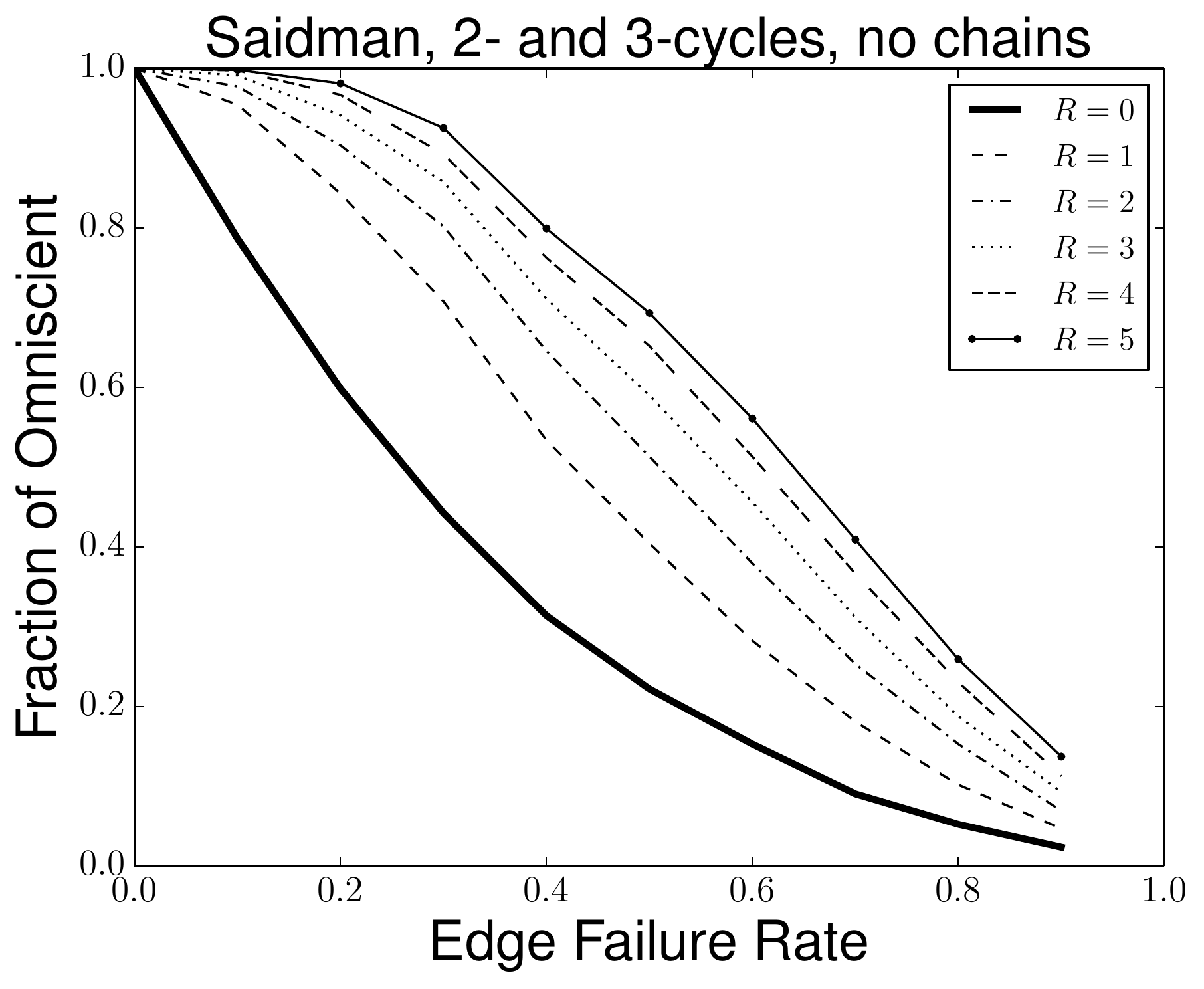}
\end{minipage}
\caption{Saidman generator graphs constrained to \num{2}-cycles only (left) and both \num{2}- and \num{3}-cycles (right).}
\label{fig:experiments-saidman}
\end{figure}

The utility of even a small number of edge queries is evident in Figure~\ref{fig:experiments-saidman}.  Just a single round of testing ($R=1$) results in \num{50.6}\% of omniscient---compared to just \num{29.8}\% with no edge testing---for edge failure probability $f = 0.5$ in the \num{2}-cycle case, and there are similar gains in the \num{2}- and \num{3}-cycle case.  For the same failure rate, setting $R = 5$ captures \num{84.0}\% of the omnsicient \num{2}-cycle matching and \num{69.3}\% in the \num{2}- and \num{3}-cycle case---compared to just \num{22.2}\% when no edges are queried.  Interestingly, we found no statistical difference between non-adaptive and adaptive matching on these graphs.

\subsection{Experiments on real match runs from the UNOS nationwide kidney exchange}
\label{sec:experiments-unos}

We now analyze the effect of querying a small number of edges per vertex on graphs drawn from the real world.  Specifically, we use the first \num{169} match runs of the UNOS nationwide kidney exchange, which began matching in October 2010 on a monthly basis and now includes \num{143} transplant centers---that is, \num{60}\% of the centers in the U.S.---and performs match runs twice per week.  These graphs, as with other fielded kidney exchanges~\cite{AJM13}, are substantially less dense than those produced by the Saidman generator.  This disparity between generated and real graphs has led to different theoretical results (e.g., efficient matching does not require long chains in a deterministic dense model~\cite{AR13,DPS12} but does in a sparse model~\cite{AGRR11}) and empirical results (both in terms of match composition and experimental tractability~\cite{Constantino13:New,Glorie14:Kidney,Anderson15:Finding}) in the past---a trend that continues here.

Figure~\ref{fig:experiments-unos2-without-zero} shows the fraction of the omniscient \num{2}-cycle and \num{2}-cycle with chains match size achieved by using only \num{2}-cycles or both \num{2}-cycles and chains and some small number of non-adaptive edge query rounds $R \in \{0,1,\ldots,5\}$.  For each of the \num{169} pre-test compatibility graphs and each of edge failure rates, \num{50} different ground truth compatibility graphs were generated.  Chains can partially execute; that is, if the third edge in a chain of length $3$ fails, then we include all successful edges (in this case, $2$ edges) until that point in the final matching.  More of the omniscient matching is achieved (even for the $R=0$ case) on these real-world graphs than on those from the Saidman generator presented in Section~\ref{sec:experiments-saidman}.  Still, the gain realized even by a small number of edge query rounds is stark, with $R=5$ achieving over \num{90}\% of the omniscient objective for every failure rate in the \num{2}-cycles-only case, and over \num{75}\% of the omniscient objective when chains are included (and typically much more).

\begin{figure}[ht!bp]
\centering
\begin{minipage}[b]{0.48\textwidth}
\centering
\includegraphics[width=\textwidth]{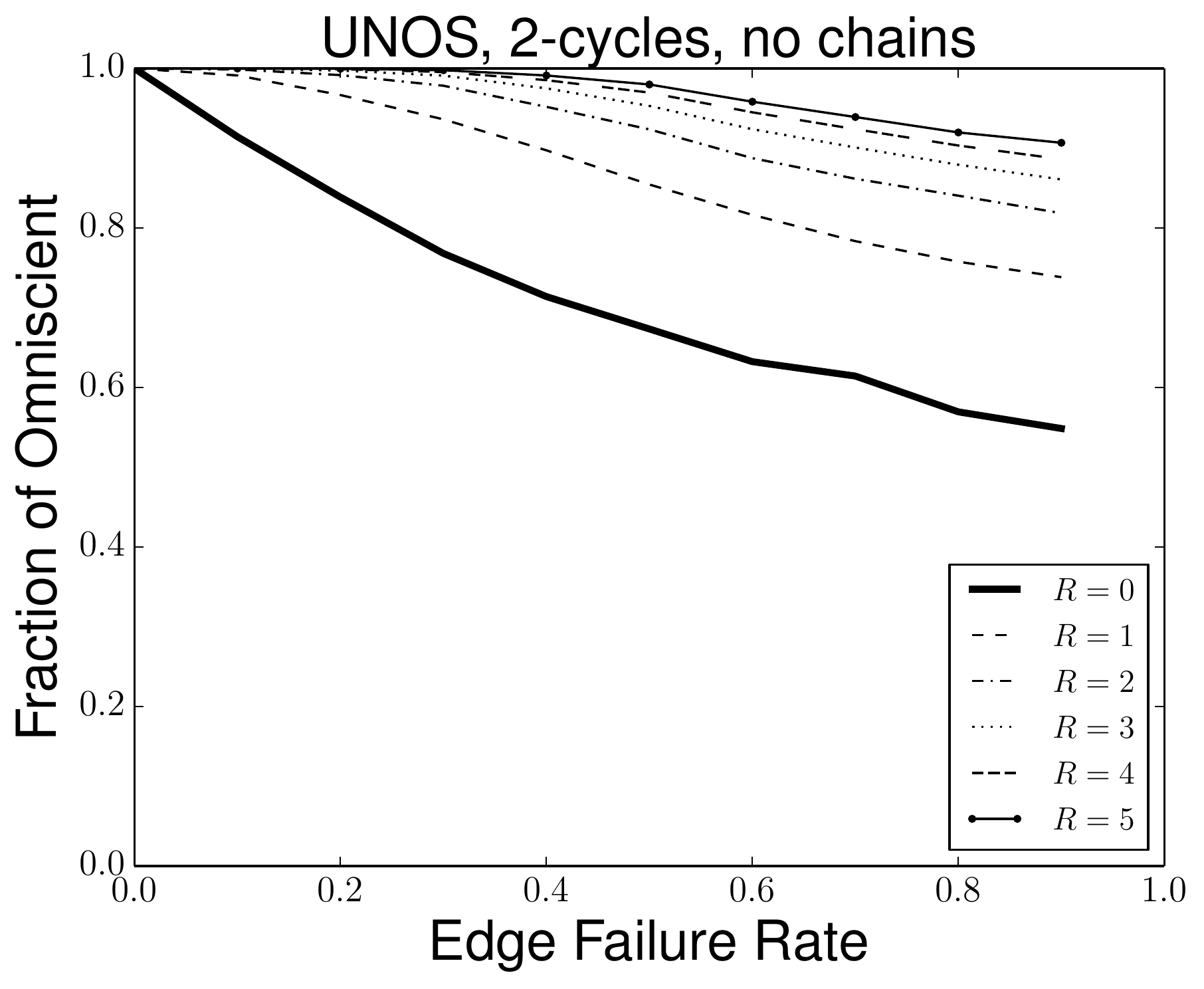}
\end{minipage}%
\hfill%
\begin{minipage}[b]{0.48\textwidth}
\centering
\includegraphics[width=\textwidth]{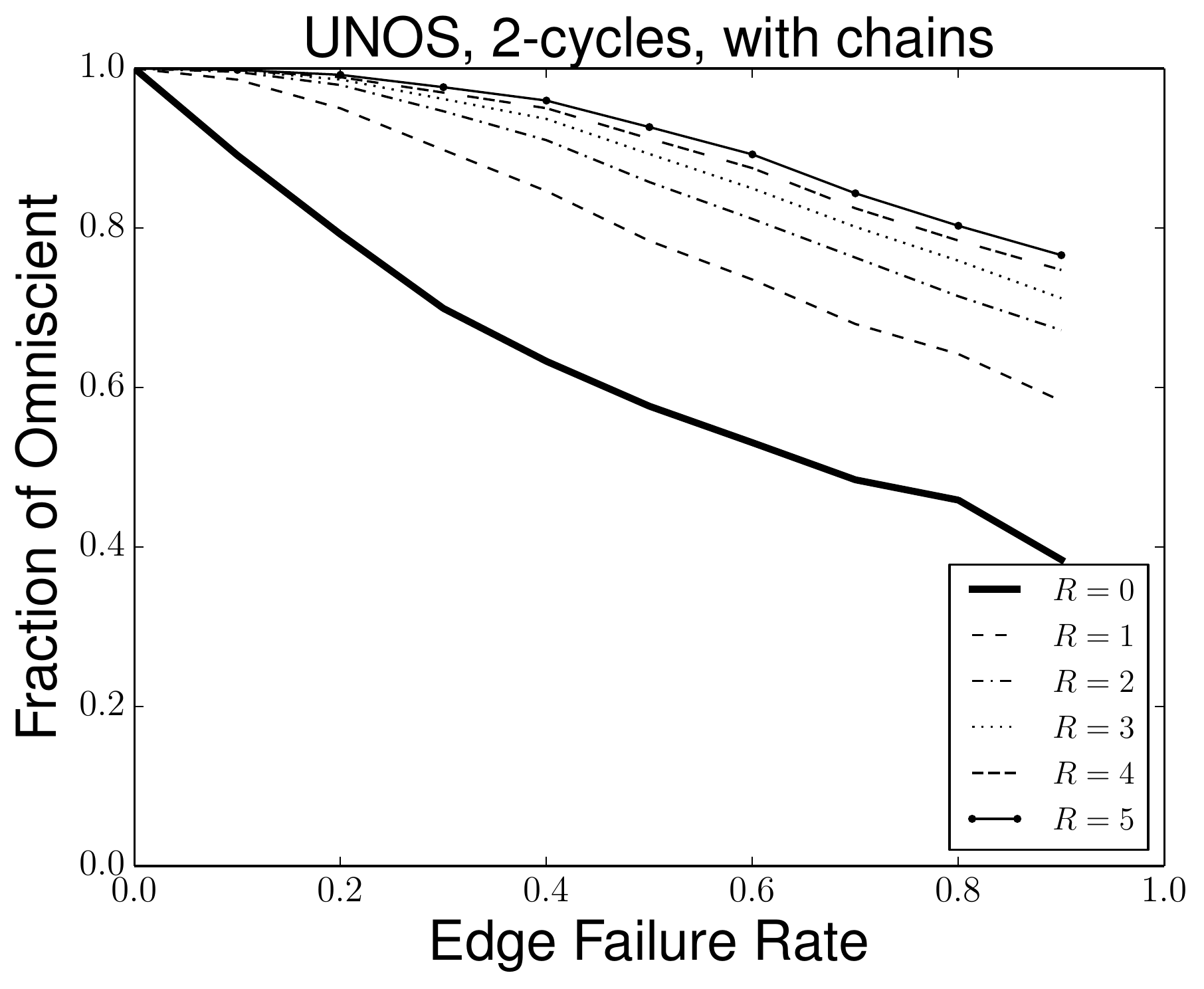}
\end{minipage}
\caption{Real UNOS match runs constrained to \num{2}-cycles (left) and both \num{2}-cycles and chains (right).}
\label{fig:experiments-unos2-without-zero}
\end{figure}

Figure~\ref{fig:experiments-unos3-without-zero} expands these results to the case with \num{2}- and \num{3}-cycles, both without and with chains.  Slightly less of the omniscient matching objective is achieved across the board, but the overall increases due to $R \in \{1,\ldots,5\}$ non-adaptive rounds of testing is once again prominent.  Interestingly, we did not see a significant difference in results for adaptive and non-adaptive edge testing on the UNOS family of graphs, either.

\begin{figure}[ht!bp]
\centering
\begin{minipage}[b]{0.48\textwidth}
\centering
\includegraphics[width=\textwidth]{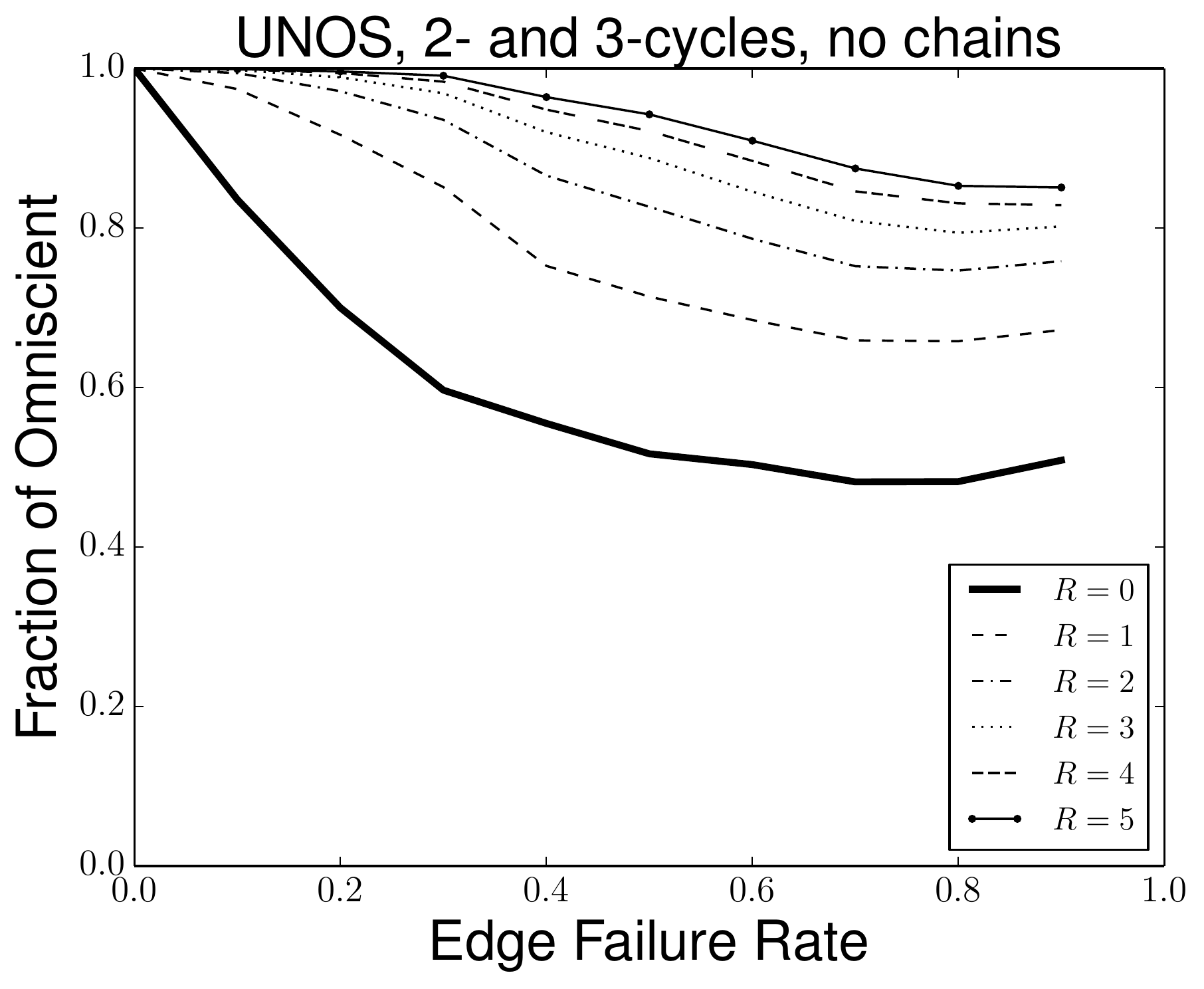}
\end{minipage}%
\hfill%
\begin{minipage}[b]{0.48\textwidth}
\centering
\includegraphics[width=\textwidth]{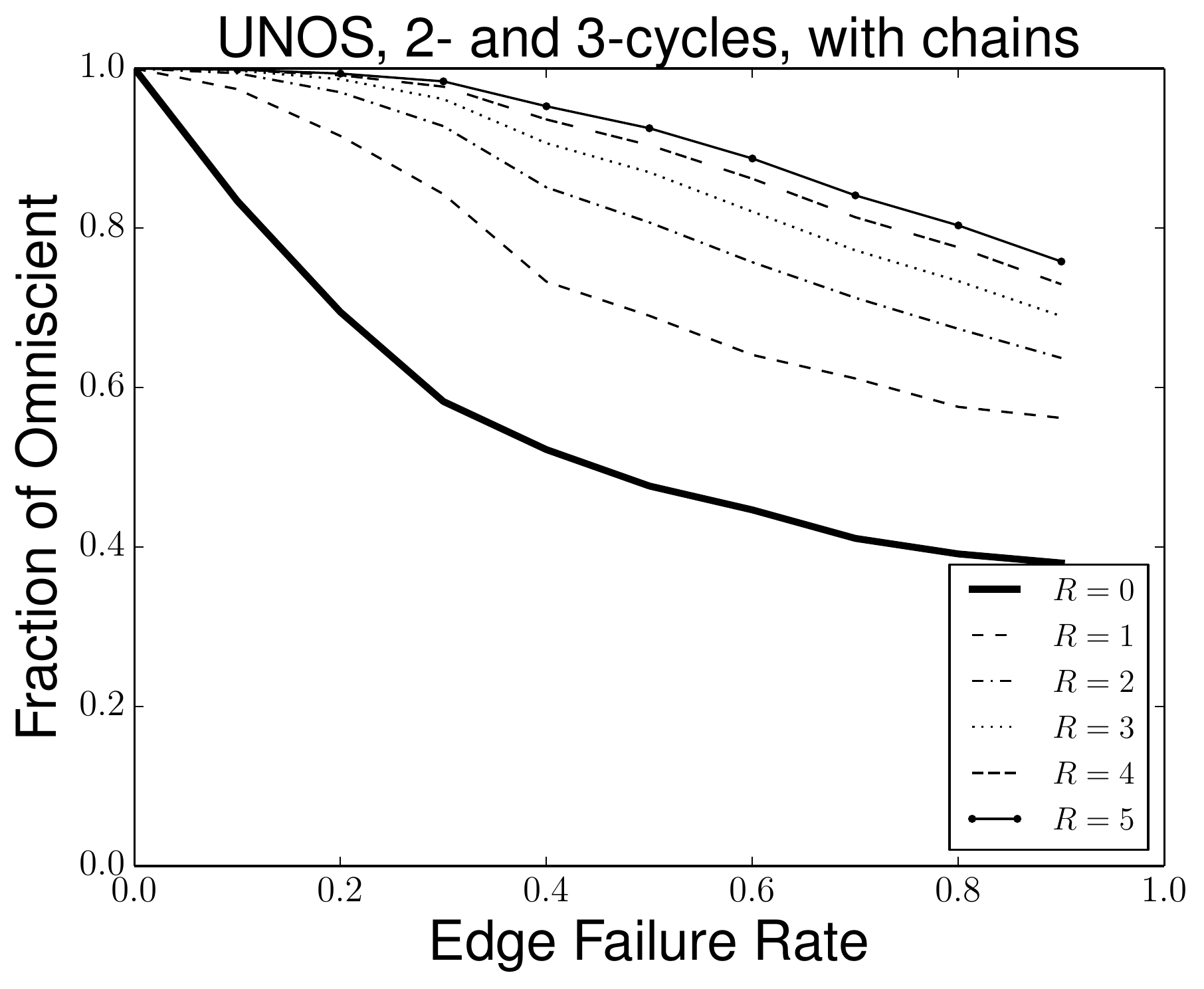}
\end{minipage}
\caption{Real UNOS match runs with \num{2}- and \num{3}-cycles and no chains (left) and with chains (right).}
\label{fig:experiments-unos3-without-zero}
\end{figure}

We provide additional experimental results in Appendix~\ref{app:experiments-appendix}.  Code to replicate all experiments is available at \url{https://github.com/JohnDickerson/KidneyExchange}; this codebase includes graph generators but, due to privacy concerns, does not include the real match runs from the UNOS exchange.

\section{Conclusions \& future research}
\label{sec:conclusions}

In this paper, we addressed stochastic matching and its generalization to $k$-set packing from both a theoretical and experimental point of view.  For the stochastic matching problem, we designed an \emph{adaptive} algorithm that queries only a constant number of edges per vertex and achieves a $(1-\epsilon)$ fraction of the omniscient solution, for an arbitrarily small $\epsilon>0$---and performs the queries in only a constant number of rounds. We complemented this result with a \emph{non-adaptive} algorithm that achieves a $(0.5 - \epsilon)$ fraction of the omniscient optimum.

We then extended our results to the more general problem of \emph{stochastic $k$-set packing} by designing an adaptive algorithm that achieves a $(\frac{2}{k} - \epsilon)$ fraction of the omniscient optimal solution, again with only $O(1)$ queries per element. This guarantee is quite close to the best known polynomial-time approximation ratio of $\frac{3}{k+1} -\epsilon$ for the standard \emph{non-stochastic} setting~\citep{Furer:2013}.

We adapted these algorithms to the kidney exchange problem and, on both generated and real data from the first \num{169} runs of the UNOS US nationwide kidney exchange, explored the effect of a small number of edge query rounds on matching performance.  In both cases---but especially on the real data---a very small number of non-adaptive edge queries per donor-patient pair results in large gains in expected successful matches across a wide range of edge failure probabilities.  

\subsection{Open theoretical problems}
\label{sec:conclusions-open}
\red{Three} main open theoretical problems remain open. First, our \emph{adaptive} algorithm for the matching setting achieves a $(1-\epsilon)$-approximation in $O(1)$ rounds and using $O(1)$ queries per vertex. Is there a \emph{non-adaptive algorithm} that achieves the same guarantee? Such an algorithm would make the practical message of the theoretical results even more appealing: instead of changing the \emph{status quo} in two ways---more rounds of crossmatch tests, more tests per patient---we would only need to change it in the latter way. 

\red{
Second, for both our adaptive and non-adaptive algorithms, the number of rounds ($R=\tilde O(p^{-1/\epsilon})$) --- even though  independent of the number of donor-patient pairs --- is exponential in $\frac 1\epsilon$.
On the other hand, our experiments show gains as high as $ 85\%$ for even small values of $R\leq 5$. This leaves open an interesting question regarding the dependence of $R$ on the values of $p$ and $\epsilon$. Can a similar $1-\epsilon$ guarantee for general graphs be obtained using a number of rounds with better dependence on $p$ and $\epsilon$? If not, are there structural properties of kidney exchange graphs that we can exploit to achieve theoretical results with better dependence of $R$ on $p$ and $\epsilon$?
}

Third, for the case of $k$-set packing, we achieve a $(\frac{2}{k}-\epsilon)$-approximation using $O(n)$ queries---in polynomial time. In kidney exchange, however, our scarcest resource is crossmatch tests; computational hardness is circumvented daily, through integer programming techniques~\citep{ABS07}. Is there an exponential-time adaptive algorithm for $k$-set packing that requires $O(1)$ rounds and $O(n)$ queries, and achieves a $(1-\epsilon)$-approximation to the omniscient optimum? A positive answer would require a new approach, because ours is inherently constrained to constant-size augmenting structures, which cannot yield an approximation ratio better than $\frac{2}{k}-\epsilon$, even if we could compute optimal solutions to $k$-set packing~\citep{Hurkens:1989}. 

\subsection{Discussion of policy implications of experimental results}
\label{sec:conclusions-policy}

Policy decisions in kidney exchange have been linked to economic and computational studies since before the first large-scale exchange was fielded in 2003--2004~\cite{RSU04,RSU05}.  A feedback loop exists between the reality of fielded exchanges---now not only in the United States but internationally as well---and the theoretical and empirical models that inform their operation, such that the latter has grown substantially closer to accurately representing the former in recent years.  That said, many gaps still exist between the mathematical models used in kidney exchange studies and the systems that actually provide matches on a day-to-day basis.

More accurate models are often not adopted quickly, if at all, by exchanges.  One reason for this is complexity---and not in the computational sense.  Humans---doctors, lawyers, and other policymakers who are not necessarily versed in optimization or theoretical economics and computer science---and the organizations they represent rightfully wish to understand the workings of an exchange's matching policy.  The techniques described in this paper are particularly exciting in that they are quite easy to explain in accessible language \red{and they involve only mild changes to the status quo}.  At a high level, we are proposing to test some small number of promising potential matches for some subset of patient-donor pairs in a pool.  As Section~\ref{sec:experiments-unos} shows, even a \emph{single} extra edge test per pair will produce substantially better results. 

\red{
Any new policy for kidney exchange has to address three practical restrictions in this space: (i) the monetary cost of crossmatches, (ii) the number of
crossmatches that can be performed per person, as there is an inherent limit on the amount of blood that can be drawn from a person, and (iii) the time it takes to find the matches, as time plays a major role in  the health of  patients and crossmatches become less accurate as time passes and the results get old.
For both our non-adaptive and adaptive algorithms, even a very small number of rounds ($R \leq 5$) results in a very large gain in the objective.  This is easily within the limits of considerations (i) and (ii) above.
Our non-adaptive algorithm performs all chosen crossmatches in parallel, so the time taken by this method is similar to the current approach. 
Our adaptive algorithm, in practice, can be implemented by a one-time retrieval of $R$ rounds worth of blood from each donor-patient pair, then sending that blood to a central wet laboratory. 
Most crossmatches are performed via an ``immediate spin'', where the bloods are mixed together and either coagulate (which is bad) or do not (which is good).  These tests are very fast, so a small number of rounds could be performed in a single day (assuming that tests in the same round are performed in parallel). 
Therefore, the timing constraint (iii) is not an issue for small $R$ (such as that used in our experiments) for the adaptive algorithm.
}

Clearly, more extensive studies would need to be undertaken before an exact policy recommendation could be made.  These studies could take factors like the monetary cost of an extra crossmatch test or variability in testing prowess across different medical laboratories into account explicitly during the optimization process.  Furthermore, various prioritization schemes could be implemented to help, for example, hard-to-match pairs find a feasible match by assigning them a higher edge query budget than easier-to-match pairs.  The positive theoretical results presented in this paper, combined with the promising experimental results on real data, provide a firm basis and motivation for this type of policy analysis.

\bibliographystyle{acmsmall}
\bibliography{abb,ultimate,matching-bib,matching-jpd}

\appendix
\medskip
\section{Additional Na\"ive  Algorithm and its Performance}
\label{app:examples}

Consider a non-adaptive algorithm which queries
$o(n)$ random neighbors of each vertex.  
The following example shows that this algorithm performs poorly.

\tikzset{
ell/.style={draw,ellipse,minimum height=2em,minimum width=5em,align=center},
line/.style={-,line width=0.6pt},
dot/.style = {draw,fill,circle,inner sep=0pt,outer sep=0pt,minimum size=3pt},
}
\begin{figure}[H]
\centering
\begin{tikzpicture}
\node[dot](B1) {};
\node[dot, below=0.7cm of  B1](B2) {};
\node[dot, below=0.7cm of  B2](B4) {};
\node[dot, below= 1.4cm of B2](B3) {};

\node[dot, left=1.5cm of B2](A1) {};
\node[dot, left=1.5cm of B4](A2) {};

\node[dot, right=1.5cm of B1] (C1) {};
\node[dot, below=0.7cm of  C1](C2) {};
\node[dot, below=0.7cm of  C2](C4) {};
\node[dot, below= 1.4cm of C2](C3) {};

\node[dot, right=1.5cm of C2](D1) {};
\node[dot, right=1.5cm of C4](D2) {};

\node[ell, dashed, rotate  = 90] at ($(A1)!0.5!(A2)$) {};
\node[ell, dashed, rotate  = 90] at ($(D1)!0.5!(D2)$) {};
\node[ell, dashed, rotate  = 90, minimum width=8em] at  ($(B1)!0.5!(B3)$){};
\node[ell, dashed, rotate  = 90, minimum width=8em] at  ($(C1)!0.5!(C3)$){};

\draw[line] (B1) -- (C1);
\draw[line] (B2) -- (C2);
\draw[line] (B3) -- (C3);
\draw[line] (B4) -- (C4);

\draw[line] (A1) -- (B1);
\draw[line] (A1) -- (B2);
\draw[line] (A1) -- (B3);
\draw[line] (A1) -- (B4);
\draw[line] (A2) -- (B1);
\draw[line] (A2) -- (B2);
\draw[line] (A2) -- (B3);
\draw[line] (A2) -- (B4);

\draw[line] (D1) -- (C1);
\draw[line] (D1) -- (C2);
\draw[line] (D1) -- (C3);
\draw[line] (D1) -- (C4);
\draw[line] (D2) -- (C1);
\draw[line] (D2) -- (C2);
\draw[line] (D2) -- (C3);
\draw[line] (D2) -- (C4);

\node[below= 0.5cm of B3] (Bset) {$B$};
\node[below= 0.5cm of C3] () {$C$};
\node[below= 1.2cm of D2] (){$D$};
\node[below = 1.2cm of A2]() {$A$};

\end{tikzpicture}
\caption{Illustration of the construction in Example~\ref{ex:random}, for $t=4$ and $\beta=1/2$.}
\label{fig:counter}
\end{figure}
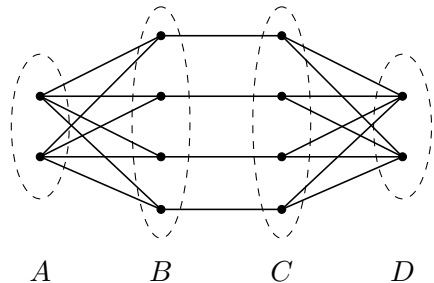

\begin{example}
\label{ex:random}
Consider the graph $G=(V,E)$ whose vertices are partitioned into sets $A$, $B$, $C$, and $D$,  such that $|A| = |D| = t^\beta$ and $|B|= |C| = t$, for some $1 >\beta >0$. Note that in this graph $n = \Theta(t)$.
Let $E$  consist of one perfect matching between vertices of $B$ and $C$, and two complete bipartite graphs, one between $A$ and $B$, and another between $C$ and $D$. See Figure~\ref{fig:counter} for an illustration. Let $p=0.5$ be the existence probability of any edge.

The omniscient optimal solution can use any edge, and, in particular, it can use the edges between $B$ and $C$. Since, these edges form a matching of size $t$ and $p=0.5$, they alone provide a matching of expected size $t/2$.
Hence, $\expmat{E}\geq t/2$.

Now, for any $\alpha<\beta$, consider the  algorithm  that queries $t^\alpha$ random neighbors for each vertex.
For every vertex in $B$, the probability that its edge to $C$ is chosen is at most $\frac{t^\alpha}{t^\beta +1}$ (similarly for the edges from $C$ to $B$).
Therefore, the expected number of edges chosen between $B$ and $C$ is at most $\frac{2t^{1+\alpha}}{t^\beta +1}$, and the expected number of existing edges between $B$ and $C$, after the coin tosses, is at most $\frac{t^{1+\alpha}}{t^\beta +1}$.
$A$ and $D$ each have $t^\beta$ vertices, so they contribute at most $2t^\beta $ edges to any matching.
Therefore, the expected size of the overall matching is no more than $t^{1+\alpha - \beta} + 2t^\beta$. Using $n=\Theta(t)$, we conclude that the approximation ratio of the na\"ive algorithm approaches $0$, as $n\rightarrow \infty$. For $\alpha  = 0.5$ and $\beta = 0.75$, the approximation ratio of the na\"ive algorithm is $O(1 / n^{0.25})$, at best. 
\end{example}

\section{Missing Proofs from Section~\ref{SEC:NON-ADAPTIVE-MATCHING} }  
\label{app:non-adap}

\subsection{Analysis of the Non-Adaptive Algorithm}

\begin{lemma}
\label{clm:exp-match-sum-parts}
Let $E_{1}$ be an arbitrary subset of edges of $E$, and let $E_{2} = E \setminus E_{1}$. Then $\expmat{E} \le \expmat{E_{1}} + \expmat{E_{2}}$.
\end{lemma}
\begin{proof}
Let $E'$ be an arbitrary subset of edges of $E$, and let $E'_{1} = E_{1} \cap E'$ and $E'_{2} = E_{2} \cap E'$. We claim that $\card{M(E')} \le \card{M(E'_{1})} + \card{M(E'_{2})}$. This is because if $T$ is the set of edges in a maximum matching in graph $(V,E')$, then clearly $T\cap E'_{1}$ and $T \cap E'_{2}$ are valid matchings in $E'_{1}$ and $E'_{2}$ respectively, and thereby it follows that $\card{M(E'_{1})} \ge \card{T\cap E'_{1}}$ and $\card{M(E'_{2})} \ge \card{T\cap E'_{2}}$, and hence $\card{M(E')} \le \card{M(E'_{1})} + \card{M(E'_{2})}$. Expectation is a convex combination of the values of the outcomes. For every subset $E'$  of edges in $E$, multiplying the above inequality by the probability that the outcome of the coin tosses on the edges of $E$ is $E'$, and then summing the various inequalities, we get $\expmat{E} \le \expmat{E_{1}} + \expmat{E_{2}}$.
\end{proof}

In order to lower bound $\expmat{W_{R}}$, we first show that for any round $r$, either our current collection of edges has an expected matching size $\expmat{W_{r-1}}$ that compares well with $\expmat{E}$, or in round $r$, we have a significant increase in $\expmat{W_{r}}$ over $\expmat{W_{r-1}}$.

\begin{lemma}
\label{lem:non-ada-per-step}
At any iteration $r\in [R]$ of Algorithm~\ref{alg:non-adaptive-matching} and odd $L$, if $\expmat{W_{r-1}} \le \expmat{E}/2$, then $$\expmat{W_{r}} \ge \frac{\alpha}{2}~\expmat{E} + (1- \gamma)\expmat{W_{r-1}},$$ where $\gamma  = p^{(L+1)/2}(1 + \frac{L+ 1}{2})$ and $\alpha = p^{(L+1)/2}$.
\end{lemma}

\begin{proof}
Define $U = E \setminus  W_{r-1}$. Assume that $\expmat{W_{r-1}}\leq \expmat{E}/2$. By Lemma~\ref{clm:exp-match-sum-parts}, we know that $\expmat{U} \ge \expmat{E} - \expmat{W_{r-1}}$. Hence, $\card{O_{r}} = \card{M(U)} \ge \expmat{U} \ge \expmat{E} - \expmat{W_{r-1}}\ge \expmat{E}/2$. 

In a thought experiment, say at the beginning of round $r$, we query the set $W_{r-1}$ and let $W'_{r-1}$ be the set of edges that are found to exist. By Lemma~\ref{lem:struct-res-matching}, there are at least $\card{O_{r}} - (1+\frac{2}{L+1})\card{M(W'_{r-1})}$ augmenting paths of length at most $L$ in $O_{r} \Delta M(W'_{r-1})$ that augment $M(W'_{r-1})$. Each of these paths succeeds with probability at least $p^{(L+1)/2}$. We have, 
\begin{align*}
\condexpmat{O_{r} \cup W'_{r-1}}{W'_{r-1}} - \card{M(W'_{r-1})}& \ge p^{(L+1)/2}~\left(\card{O_{r}} - (1+ \frac{2}{L+1})\card{M(W'_{r-1})}\right)\\
&\ge p^{(L+1)/2}~\left(\frac{1}{2}~\expmat{E}- (1+ \frac{2}{L+1})\card{M(W'_{r-1})}\right)~,
\end{align*}
where the expectation on the left hand side is taken only over the outcome of the edges in $O_{r}$.  Therefore, we have $\condexpmat{O_{r} \cup W'_{r-1}}{W'_{r-1}} \ge \frac{\alpha}{2}~\expmat{E} + (1- \gamma)\card{M(W'_{r-1})}$, where $\alpha = p^{(L+1)/2}$ and $\gamma = p^{(L+1)/2}~(1+ \frac{2}{L+1})$.
Taking expectation over the coin tosses on $W_{r-1}$ that create outcome $W'_{r-1}$, we have our result, i.e.,
\begin{align*}
\expmat{W_{r}} \geq \expc_{W_{r-1}}[\condexpmat{O_{r} \cup W'_{r-1}}{W'_{r-1}}] \geq \expmat{O_{r} \cup W_{r-1}} \ge \frac{\alpha}{2}~\expmat{E} + (1- \gamma)\expmat{W_{r-1}}.
\end{align*}
\end{proof}

\begin{proof}[\textsc{of Theorem~\ref{thm:non-adaptive-alg}}]
For ease of exposition, assume $L=\frac{4}{\epsilon}-1$ is an odd integer. Then, either $\expmat{W_{R}} \geq \expmat{E}/2$ in which case we are done. Or otherwise, by repeatedly applying Lemma~\ref{lem:non-ada-per-step} for $R$ steps, we have
\begin{align*}
\expmat{W_{R}} ~\ge~ \frac{\alpha}{2}(1+ (1-\gamma) + (1-\gamma)^{2} + \cdots + (1-\gamma)^{R-1})\expmat{E} ~\ge~ \frac{\alpha}{2}\frac{(1 - (1-\gamma)^{R})}{\gamma}\expmat{E}.
\end{align*}
Now, $\frac{\alpha}{\gamma}(1 - (1-\gamma)^{R}) \ge 1 - \frac{2}{L+1} - e^{-\gamma R} \ge 1 - \epsilon$ for $R = \frac{\log (2/ \epsilon)}{p^{2/\epsilon}}$.
Hence, we have our $0.5 (1-\epsilon)$ approximation.
\end{proof}

\subsection{Example Graph for the Non-Adaptive Algorithm}
\begin{lemma}
\label{clm:perfect-bipartite}
Let $G=(U \cup V, U \times V)$ be a complete bipartite graph between $U$ and $V$ with $\card{U} = \card{V} = n$. For any constant probability $p$, $\expmat{E} \ge n - o(n)$.
\end{lemma}
\begin{proof}
Denote by $E_{p}$ the random set of edges formed by including each edge in $U \times V$ independently with probability $p$. We show that with probability at least $1 - \frac{1}{n^{8}}$, over the draw $E_{p}$, the maximum matching in the graph $(U \cup V, E_{p})$ is at least $n  - c \log(n)$, where $c = 10/\log(\frac{1}{(1-p)})$, and this will complete our claim.

In order to show this, we prove that with probability at least $1 - \frac{1}{n^{8}}$, over the draw $E_{p}$, all subsets $S \subseteq U$ of size at most $n  - c \log(n)$, have a neighborhood of size at least $\card{S}$. By Hall's theorem, our claim will follow.

Consider any set $S \subseteq U$ of size at most $n - c \log(n)$. We will call set $S$ `bad' if there exists some set $T\subseteq V$ of size $(\card{S} - 1)$ such that $S$ does not have edges to $V \setminus T$.  Fix any set $T \subseteq V$ of size $\card{S}-1$. Over draws of $E_{p}$, the probability that $S$ has no outgoing edges to $V \setminus T$ is at most $(1-p)^{\card{S}\card{V\setminus T}} = (1-p)^{\card{S}(n - \card{S} + 1)}$. Hence, by union bound, the probability  that $S$ is bad is at most ${n \choose \card{S} - 1} (1-p)^{\card{S}(n - \card{S} + 1)}$.

Again, by union bound, the probability that some set $S\subseteq U$ of size at most $n - c \log(n)$ is bad is at most $ \sum_{1 \le \card{S} \le n - c \log(n)}{n \choose \card{S}}{n \choose \card{S} - 1} (1-p)^{\card{S}(n - \card{S} + 1)}$ and this in turn is at most
{\small
\begin{align*}
\sum_{1 \le \card{S} \le n - c\log(n)}  n^{\card{S}} n^{\card{S}}  (1-p)^{\card{S}(n - \card{S} + 1)} \le \sum_{1 \le \card{S} \le n - c\log(n)} e^{\card{S}\cdot (2 \log(n) + (n + 1) \log(1-p) - \card{S} \log(1-p))}
\end{align*}
}

Note that the exponent in the summation achieves its maximum for $|S| = 1$. For $c = 10/\log(\frac{1}{1-p})$, we have that the given sum is at most $\exp(-\frac{n}{2}\log(\frac{1}{1-p}))$, and hence with high probability, no set $S \subseteq U$ of size at most $n - c\log(n)$ is bad.
\end{proof}

\begin{proof}[of Theorem~\ref{thm:non-adaptive-upperbound}]
Let $(V,E)$ be a graph, illustrated in Figure \ref{fig:upper-bound}, whose vertices are partitioned into sets $A$, $B$, $C$, and $D$,  such that $|A| = |D| = \frac{t}{2}$, $|B|= |C| = t$. The edge set $E$ consists  of one perfect matching between vertices of $B$ and $C$, and two complete bipartite graphs, one between $A$ and $B$, and another between $C$ and $D$. Let $p=0.5$ be the existence probability of any edge.

We first examine the value of the omniscient optimal, $\expmat{E}$. Since $p=0.5$, in expectation, half of the edges in the perfect matching between $B$ and $C$ exist, and therefore half of the vertices of $B$ and $C$ will get matched. By Lemma~\ref{clm:perfect-bipartite}, with high probability, the complete bipartite graph between the remaining half of $B$ and $A$ has a matching of size at least $t/2 - o(t)$. And similarly, with high probability, the complete bipartite graph between remaining half of $C$ and $D$ has a matching of size at least $t/2 - o(t)$. Therefore, $\expmat{E}$ is at least $\frac{3}{2}t - o(t)$.

Next, we look at Algorithm~\ref{alg:non-adaptive-matching}. For ease of exposition, let $B_1$  and $B_2$ denote the top and bottom half of the vertices in $B$. Similarly, define $C_1$ and $C_2$. Since Algorithm~\ref{alg:non-adaptive-matching} picks maximum matchings arbitrarily, we show that there exists a way of picking maximum matchings such that the expected matching size of the union of the edges picked in the matching is at most $\frac{5}{4}t\,\,$ ($=\frac{5}{6}\, \frac{3}{2}t$).

Consider the following choice of maximum matching picked by the algorithm: In the first round, the algorithm picks the perfect matching between $B_{1}$ and $C_{1}$, and a perfect matching between $A$ and $B_{2}$, and a perfect matching between $C_{2}$ and $D$. In the second round, the algorithm picks the perfect matching between $B_{2}$ and $C_{2}$, and a perfect matching each between $A$ and $B_{1}$, and between $C_{1}$ and $D$. After these two rounds, we can see that there are no more edges left between $B$ and $C$. For the subsequent $R-2$ rounds, in each round, the algorithms picks a perfect matching between $A$ and $B_{1}$, and a perfect matching between $C_{1}$ and $D$. It is easy to verify that in every round, the algorithm has picked a maximum matching from the remnant graph.

We analyze the expected size of matching output by the algorithm. For each of the vertices in $B_{2}$ and $C_{2}$, the algorithm has picked only two incident edges. For any vertex in $B_{2}$ and $C_{2}$, with probability at least $(1-p)^{2} = \frac{1}{4}$, none of these two incident edges exist. Hence, the expected number of vertices that are \emph{unmatched} in $B_{2}$ and $C_{2}$ is at least $\frac{1}{4}(\frac{t}{2} + \frac{t}{2}) = \frac{t}{4}$. Since the vertices in $A$ can only be matched with vertices in $B$, and the vertices in $D$ can only be matched with vertices in $C$, it follows that at least $\frac{t}{4}$ of the vertices in $A$ and $C$ are unmatched in expectation. Hence, the total number of edges included in the matching is at most $\frac{5}{4}t$. This completes our claim.
\end{proof}

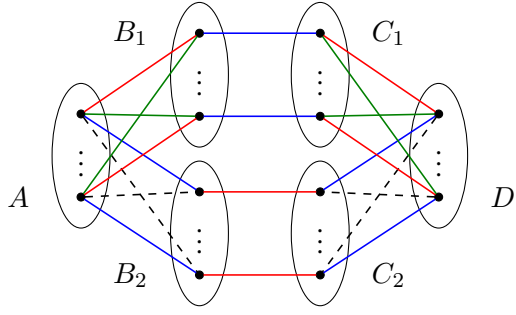
\begin{figure}
\centering
\begin{tikzpicture}
\node[dot](B1) {};
\node[below = 0.1cm of B1]{$\vdots$};
\node[dot, below=1cm of  B1](B15) {};
\node[dot, below=0.9cm of  B15](B2) {};
\node[below= 0.1cm of B2]{$\vdots$};
\node[dot, below= 1cm of B2](B3) {};

\node[dot, right=1.5cm of B1] (C1) {};
\node[below = 0.1cm of C1]{$\vdots$};
\node[dot, below=1cm of  C1](C15) {};
\node[dot, below=0.9cm of  C15](C2) {};
\node[below= 0.1cm of C2] {$\vdots$};
\node[dot, below= 1cm of C2](C3) {};

\node[dot, below left=1 and 1.5cm of B1](A1) {};
\node[below= 0.1cm of A1](Adots) {$\vdots$};
\node[dot, below= 1cm of A1](A2) {};

\node[dot, below right=1cm  and 1.5cm  of C1](D1) {};
\node[below= 0.1cm of D1](Ddots) {$\vdots$};
\node[dot, below= 1cm of D1](D2) {};

\node[ell, rotate  = 90] at (Adots) {};
\node[ell, rotate  = 90] at (Ddots) {};
\node[ell, rotate  = 90] at  ($(B1)!0.5!(B15)$) {};
\node[ell, rotate  = 90] at  ($(B2)!0.5!(B3)$){};
\node[ell, rotate  = 90] at  ($(C1)!0.5!(C15)$){};
\node[ell, rotate  = 90] at  ($(C2)!0.5!(C3)$){};

\draw[line, blue] (B1) -- (C1);
\draw[line, blue] (B15) -- (C15);
\draw[line, red] (B2) -- (C2);
\draw[line, red] (B3) -- (C3);

\draw[line, red] (A1) -- (B1);
\draw[line, blue] (A1) -- (B2);
\draw[line, green!50!black] (A1) -- (B15);
\draw[line, dashed] (A1) -- (B3);
\draw[line, green!50!black] (A2) -- (B1);
\draw[line, red] (A2) -- (B15);
\draw[line, dashed] (A2) -- (B2);
\draw[line, blue] (A2) -- (B3);

\draw[line, red] (D1) -- (C1);
\draw[line, blue] (D1) -- (C2);
\draw[line, green!50!black] (D1) -- (C15);

\draw[line, dashed] (D1) -- (C3);
\draw[line, green!50!black] (D2) -- (C1);
\draw[line, red] (D2) -- (C15);
\draw[line, dashed] (D2) -- (C2);
\draw[line, blue] (D2) -- (C3);

\node[left= 0.5cm of B1] (Bset) {$B_{1}$};
\node[left= 0.5cm of B3] (Bset) {$B_{2}$};
\node[right= 0.5cm of C1] (Bset) {$C_{1}$};
\node[right= 0.5cm of C3] (Bset) {$C_{2}$};
\node[right= 0.5cm of D2] (){$D$};
\node[left = 0.5cm of A2]() {$A$};

\end{tikzpicture}
\caption{ The blue and red edges represent the matching picked at rounds $1$ and $2$, respectively. The green edges represent the edges picked at round $3$ and above. The dashed edges are never picked by the algorithm. }
\label{fig:upper-bound}
\end{figure}

\section{Missing Proofs from Section~\ref{SEC:KSET}} 
\label{app:kset}
In this section, we fill in the missing proofs for stochastic $k$-set packing. A notation that we will use in some parts of the analysis is $\condexpset{A}{B}$ that we define as follows: Given a collection $B \subseteq A$ that has been queried and $B' \subseteq B$ that exists, we use $\condexpset{A}{B}$ to denote $\expc[\card{\ksize{X_{p}\cup B'}}]$ where $X_{p}$ is the random set formed by including every element of $A \setminus B$ independently with probability $p$.


\subsection{Adaptive Algorithm for $k$-Set Packing}
\label{sec:ada-kset}

We introduce some notation that is used in the remainder of the proofs in this section. At the beginning of the $r^{th}$ iteration of Algorithm~\ref{alg:adaptive-kset},  we know the results of the queries $\bigcup_{i = 1}^{r-1}Q_{i}$. We define $Z_{r}$ to be the expected cardinality of the instance $(U,A)$ given the result of these queries. More formally, $Z_{r} = \condexpset{A}{\bigcup_{i =1}^{r-1}Q_{i}}$.
We note that $Z_{1} = \expset{A}$.

For a given $r$, we use the notation $\expc_{Q_r}[X]$ to denote the expected value of $X$ where the expectation is taken over \emph{only} the outcome of query $Q_r$, and fixing the outcomes on the results of queries $\bigcup_{i=1}^{r-1}Q_i$. Moreover, for a given $r$, we use $\expc_{Q_r, \dots, Q_R}[X]$ to denote the expected value of $X$ with the expectation taken over the outcomes of queries $\bigcup_{i=r}^{R}Q_i$, and fixing an outcome on the results of queries $\bigcup_{i=1}^{r-1}Q_i$.

The next result, Lemma~\ref{lem:increase-each-iter-kset}, proves a lower bound on the expected increase in the cardinality of $B_{r}$ (the solution at round $r$) with respect to $B_{r-1}$ (the solution in the previous round).
\begin{lemma}
\label{lem:increase-each-iter-kset}
For every $r\in [R]$, it is the case that  $\expc_{Q_{r}}[\card{B_{r}}] \geq (1-\gamma) \card{B_{r-1}} + \gamma (\setapp) Z_{r}$, where $\gamma = \frac{p^{\saize}}{(\setapp) k\,\saize}$.
\end{lemma}
\begin{proof}
By Lemma~\ref{lem:struct-res-k-set}, $Q_{r}$ is a collection of at least $\frac{1}{k\,\saize}(\card{\ksize{A_{r}}} - \frac{\card{B_{r-1}}}{\setapp})$ \emph{disjoint $\saize$-size augmenting structures} $(C,D)$ for $B_{r-1}$. Since in each augmenting structure $(C,D)$, $C$ has at most $\saize$ sets, on querying, the set $C$ exists with probability at least $p^{\saize}$. Therefore, the expected increase in the size of the solution at Step~\ref{step:increment-B} is:
\begin{align*}
\expc_{Q_{r}}[\card{B_r} ] - \card{B_{r-1}} \geq p^{k\saize}~\card{Q_{r}} &\geq 
\frac{p^{\saize}}{k\,\saize}\left(\card{\ksize{A_{r}}} - \frac{\card{B_{r-1}}}{\setapp}\right) \\
  & \geq \gamma \big( (\setapp)~\card{\ksize{A_{r}}} - \card{B_{r-1}} \big).
\end{align*}
Noting that $\card{\ksize{A_{r}}} \ge Z_{r}$, we have our result.
\end{proof}

\begin{proof}[\textsc{of Theorem~\ref{thm:main-adaptive-kset}}]
First, we make a technical observation about $Z_{r}$:
For every $r\leq R$, $\expc_{Q_{r-1}} [ Z_{r} ] = Z_{r-1}$. This is since 
\begin{align}
\label{eq:conditional-expectation-kset}
\expc_{Q_{r-1}} [ Z_{r} ] &= \expc_{Q_{r-1}}[\condexpset{A}{\bigcup_{i=1}^{r-1}Q_{i}}] = \condexpset{A}{\bigcup_{i =1}^{r-2}Q_{i}}  = Z_{r-1}.
\end{align}

Now, similar to the proof of Theorem~\ref{thm:main-adaptive-matching}, we first apply Lemma~\ref{lem:increase-each-iter-kset} to the $R^{th}$ step and get $\expc_{Q_{R}}[\card{B_{R}}] \geq (1-\gamma) \card{B_{R-1}} + \gamma (\setapp) Z_{R}.$
Next taking expectation on both sides with respect to $Q_{R-1}$, we get
$\expc_{Q_{R-1}, Q_{R}}[\card{B_{R}}] \geq (1-\gamma) \expc_{Q_{R-1}}[\card{B_{R-1}}] + \gamma (\setapp) \expc_{Q_{R-1}}[Z_{R}].$
Applying Lemma~\ref{lem:increase-each-iter-kset} to $\expc_{Q_{R-1}}[\card{B_{R-1}}]$ and Equation~\eqref{eq:conditional-expectation-kset} to $\expc_{Q_{R-1}}[Z_{R}]$, we get
\begin{align*}
\expc_{Q_{R-1}, Q_{R}}[\card{B_{R}}] &\geq (1-\gamma) ((1-\gamma) \card{B_{R-2}} + \gamma (\setapp) Z_{R-1}) + \gamma (\setapp)~Z_{R-1}\\
&= (1-\gamma)^{2}  \card{B_{R-2}} + \gamma (\setapp) (1 + (1-\gamma))~Z_{R-1}.
\end{align*}

We can repeat the above steps, by sequentially taking expectation over $Q_{R-2}$ through $Q_{1}$, and applying Lemma~\ref{lem:increase-each-iter-kset} and Equation~\eqref{eq:conditional-expectation-kset} at each step, to achieve
\begin{align*}
\expc_{Q_{1}, \dots, Q_{R}}[\card{B_{R}}]  &\ge (1-\gamma)^{R}  \card{B_{0}} + \gamma (\setapp) (1 + (1-\gamma)+ \cdots + (1-\gamma)^{R-1})~Z_{1}\\
&\geq (\setapp) (1 - (1-\gamma)^{R})~\expset{A} \geq \frac{2}{k}(1- \frac{\eta k}{2}) (1 - e^{-\gamma R}). ~\expset{A}
\end{align*}
We complete the claim by noting that $$\frac{2}{k}(1- \frac{\eta k}{2}) (1 - e^{-\gamma R}) \geq \frac{2}{k} (1-\frac \epsilon 2) (1 - \frac{\epsilon}{2})\geq (1-\epsilon)\frac{2}{k},$$ where the penultimate inequality comes from the fact that  $\eta = \epsilon/k$ and $$R=\frac {(\setapp)~k~\saize}{p^{\saize}} \log( \frac 2\epsilon) = \frac 1\gamma \log(\frac 2\epsilon).$$ Therefore, the cardinality of $B_{R}$ in expectation is at least a $(1-\epsilon)\frac{2}{k}\expset{A}$.
\end{proof}

\subsection{Non-Adaptive Algorithm for $k$-Set Packing}
\label{sec:non-ada-kset}
To prove Theorem~\ref{thm:non-adaptive-sets}, we analyze the non-adaptive Algorithm~\ref{alg:non-adaptive-sets} from the main paper.
Before proving Theorem~\ref{thm:non-adaptive-sets}, we present a technical claim.
\begin{claim}
\label{clm:exp-k-set-sum-parts}
Let $\A_{1}\subseteq \A$ and $\A_{2} = \A \setminus \A_{1}$. Then $\expset{\A} \leq \expset{\A_{1}} + \expset{\A_{2}}$.
\end{claim}
\begin{proof}
Let $\A'$ be any subset of $\A$,  $\A'_{1} = \A_{1} \cap \A'$,  and $\A'_{2} = \A_{2} \cap \A'$. Since the  $k$-set packing of $\A'$ restricted to $\A'_1$ and $\A'_2$ are  valid $k$-set packings for these subsets, hence $\card{\ksize{\A'}} \leq \card{\ksize{\A'_{1})}} + \card{\ksize{\A'_{2}}}$.
For every $\A'\subseteq \A$, the above inequality holds.
Expectation is a linear combination of the values of the outcomes, and so this inequality also holds in expectation. That is, $\expset{\A} \leq \expset{\A_{1}} + \expset{\A_{2}}$.
\end{proof}

\begin{proof}[\textsc{of Theorem~\ref{thm:non-adaptive-sets}}]
We claim that the expected cardinality of the $k$-set solution output by Algorithm~\ref{alg:non-adaptive-sets} is at least $(1-\frac{\epsilon}{2}) \frac{(\setapp)^{2}}{1+ \setapp}\expset{\A}$. The claimed approximation will follow since $\eta = \frac{\epsilon}{2k}$. 

For ease of exposition, let $\alpha = \frac{\setapp}{1+\setapp}$, and now note that $\frac{(\setapp)^{2}}{1+ \setapp} = \alpha (\setapp) = (1-\alpha) (\setapp)^{2}$.

Assume that $\expset{\B_{R}} \leq  \alpha \cdot \expset {\A}$ (else it will be immediately follow that the expected cardinality of the $k$-set solution output by the algorithm is at least $(\setapp) \alpha \expset{A}$ and this will complete the claim). 

First, we make an observation. For each round $r \in [R]$, we have $\expset{B_{r}} \le \expset{\B_{R}} \leq  \alpha \expset {\A}$. If we denote $A_{r} = A \setminus B_{r-1}$, then it follows that $$\card{O_r} \geq (\setapp) \card{\ksize{\A_r}} \geq (\setapp) \expset{\A_r} \geq (\setapp)(\expset{\A} - \expset{B_{r-1}}) \geq (\setapp) (1- \alpha) \expset{\A}~,$$
where the first inequality follows from the fact that $O_{r}$ is $(\setapp)$-approximation to $A_{r}$, and the second inequality follows from Claim~\ref{clm:exp-k-set-sum-parts}.

We analyze the expected cardinality of the output solution $Q_{R}$ by analyzing the $R$ stages that the algorithm adopts at Steps \ref{step:non-ada-kset-stage1} and \ref{step:non-ada-kset-aug} to create solution $Q_{R}$. For this analysis, we use the following notation: For a given $r$, we use the notation $\expc_{O_r}[X]$ to denote the expected value of $X$ where the expectation is taken over \emph{only} the outcome of query $O_r$, and fixing the outcomes on the results of queries $\bigcup_{i=1}^{r-1}O_i$. Moreover, for a given $r$, we use $\expc_{O_r, \dots, O_R}[X]$ to denote the expected value of $X$ with the expectation taken over the outcomes of queries $\bigcup_{i=r}^{R}O_i$, and fixing an outcome on the results of queries $\bigcup_{i=1}^{r-1}O_i$.

In the first stage, $Q_{1}$ is assigned to the collection of $k$-sets that are found to exist in $O_{1}$. In the second stage, we try to augment $Q_{1}$ by finding augmenting structures from $O_{2}$ and querying them. By Lemma~\ref{lem:struct-res-k-set}, it finds at least $\frac{1}{k \saize} \left(\card{O_{2}} - \frac{\card{Q_{1}}}{\setapp}\right)$ disjoint augmenting structures from $O_{2}$ that have size at most $\saize$ and augment $Q_{1}$. Since each augmenting structure exists independently with probability at least $p^{\saize}$, \emph{in expectation over the outcomes of queries to $O_{2}$}, the size of $Q_{2}$, $\expc_{O_{2}}[Q_{2}]$, is at least
\begin{align*}
\card{Q_{1}} + p^{\saize}\left(\frac{1}{k\saize}\big(\card{O_{2}} - \frac{\card{Q_{1}}}{\setapp} \big) \right)
=& \frac{p^{\saize}}{k\saize}~\card{O_{2}} + (1 - \frac{p^{\saize}}{k\saize(\setapp)}  ) \card{Q_{1}} \\
 \ge& \frac{p^{\saize}}{k\saize}~ (\setapp) (1- \alpha)\expset{\A} + (1 - \frac{p^{\saize}}{k\saize(\setapp)}  )~ \card{Q_{1}},
\end{align*}
and hence the expected size of $Q_{2}$ is at least $\beta~\expset{A} + (1- \gamma) \card{Q_{1}}$, where $\beta = \frac{p^{\saize}}{k\saize}(\setapp) (1- \alpha)$ and $\gamma =  \frac{p^{\saize}}{k~\saize~(\setapp)}$.

For the third stage, a similar analysis shows that the expected size of $Q_{3}$, $\expc_{O_{3}}[Q_{3}]$, \emph{with expectation taken only over the outcomes of the queries to $O_{3}$}, is at least $\beta~\expset{A} + (1- \gamma) \card{Q_{2}}$. If we now, in addition, take expectation over the outcomes of queries to $O_{2}$, 
we get the expected size of $Q_{3}$, $\expc_{O_{2}, O_{3}}[Q_{3}]$, is at least $\beta~\expset{A} + (1-\gamma) ~ (\beta~\expset{A} + (1- \gamma) \card{Q_{1}}) = \beta(1 + (1-\gamma)) ~\expset{A} + (1- \gamma)^{2} ~ \card{Q_{1}}$.

Repeating the above steps, the procedure creates the $k$-set solution $Q_{R}$ (from $O_{1}, \cdots, O_{R}$) whose expected size, with expectation taken over the outcomes of queries to $O_{2}$ through $O_{R}$, is at least 
$$\beta(1+(1-\gamma) + \cdots + (1-\gamma)^{R-2}) \expset{A} + (1-\gamma)^{R-1}\card{Q_{1}}~.$$

Finally, taking expectation over outcomes of queries to $O_{1}$, since the expected size of $\card{Q_{1}}$ is at least $p \card{O_{1}} \ge p ~(\setapp) \expset{A} \ge \beta~\expset{A}$, we have that the expected size of $Q_{R}$ is at least 
\begin{align*}
\beta~(1 &+(1-\gamma) + \cdots + (1-\gamma)^{R-1}) \expset{A}\\
& =\frac{\beta}{\gamma}(1 - (1-\gamma)^{R}) \expset{A} \ge \frac{\beta}{\gamma} (1 - e^{-\gamma R})  \expset{A} \geq (1-\frac \epsilon 2)\frac{(\setapp)^2}{\setapp + 1} \expset{\A}
\end{align*}
\end{proof}

\section{Matching Under Correlated Edge Probabilities}
\label{sec:extensions}
In this section, we extend our framework to a more general setting. Here, the existence probability of an edge depends on parameters that are associated with the endpoints of the edge. Specifically, every vertex $v_i\in V$ is associated with parameter $p_{i}$, and an edge $e_{ij}=(v_i,v_j)$ exists with probability $p_ip_j$. 

Importantly, this model is a generalization of the model studied above: we can still think of each edge $e\in E$ as existing with a given probability, and these events are \emph{independent}. However, using vertex parameters gives us a formal framework for correlating the probabilities of edges incident to any particular vertex. The motivation for this comes from kidney exchange: Some \emph{highly sensitized} patients are less likely than other patients to be compatible with potential donors. Such patients correspond to a small $p_i$ parameter.

We consider two settings: adversarial and stochastic. In the adversarial setting, the vertex parameters $p_{i}$ are selected by an adversary, whereas in the stochastic model, the parameters are drawn from a distribution. 
In the former setting, for $\delta>0$, define $f_\delta$ to be the \emph{number of vertices} that have $p_{i} < \delta$. In the latter setting, for a distribution $D$ and $\delta>0$, let $g_\delta$ indicate the \emph{probability} that a vertex has its parameter less than $\delta$, i.e., $g_\delta = \Pr_{p_i\sim D}[p_i < \delta]$. We formulate our results in terms of $\delta$, $f_\delta$, and $g_\delta$, and the desired value of $\delta$ can depend on the application. For example, in kidney exchange, $\delta$ would be the probability that a highly-sensitized patient is compatible with a random donor (a patient is typically considered to be highly sensitized when this probability is $0.2$), and $f_\delta$ would be the number of highly-sensitized patients in the kidney exchange pool.

\subsection{Adaptive Algorithm in Adversarial Setting}
In this section, we consider the case where an adversary chooses the values of vertex parameters. We give guarantees on the performance of Algorithm~\ref{alg:adaptive-matching} in this setting.

\begin{theorem}\label{thm:adaptive-adversary}
For any graph $(V, E)$, any $\epsilon>0$, and $\delta>0$, Algorithm~\ref{alg:adaptive-matching} returns a matching with expected size of  $(1-\epsilon)(\expmat{E} - f_\delta)$ in $R = \frac{\log(2/\epsilon)}{\delta^{4/\epsilon}}$ iterations. 
\end{theorem}

The proof of this theorem and the subsequent lemmas are similar to the proofs of Section~\ref{sec:ada-matching}, and are included here for completeness. In the next lemma, $\expc_{Q_r}[|M_r|]$ indicates the expected size of $M_r$, where the expectation is over the query outcome of $Q_r$. More formally, $\expc_{Q_r}[|M_r|] = \condexpmat{\bigcup_{j=1}^{r} Q_j}{\bigcup_{j=1}^{r-1} Q_j}$.
We use $Z_{r}$ to denote the expected size of the maximum matching in graph $(V,E)$ given the results of the queries $\bigcup_{j = 1}^{r-1}Q_{j}$. More formally, $Z_{r} = \condexpmat{E}{\bigcup_{j = 1}^{r-1}Q_{i}}$. Note that $Z_{1} = \expmat{E}$.
\begin{lemma}
 For all $r\in [R]$ and odd $L$, 
$\expc_{Q_{r}}[|M_r|] \geq (1- \gamma) |M_{r-1}| + \alpha (Z_{r} - f_\delta),
$ where  $\gamma = \delta^{L+1}(1+\frac{2}{L+1})$ and $\alpha = \delta^{L+1}$.
\end{lemma}
\begin{proof}
By Lemma~\ref{lem:struct-res-matching}, there exists $|O_r| - (1+\frac{2}{L+1})|M_{r-1}|$ many augmenting paths in $O_r\otimes M_{r-1}$ that augment $M_{r-1}$  and have length at most $L$. These augmenting paths are disjoint, so at most $f_\delta$ of them include a vertex $v_i$, with $p_i\leq \delta$.
We will ignore these paths. Among the remaining augmenting paths, each path of length $L$, has at most $\frac{L+1}{2}$ edges that have not been queried yet. These edges do not share a vertex, so each one exists, independently of others, with probability at least $\delta^2$. Therefore, 
the expected increase in the size of the matching from these augmenting paths is:
\[\expc_{Q_r}[|M_r|] - |M_{r-1}| \geq \delta^{L+1}\left( |O_r| - (1+\frac{2}{L+1})|M_{r-1}| - f_\delta \right) \geq \alpha (Z_r - f_\delta)  - \gamma |M_{r-1}|.
\]
where the last inequality holds by the fact that $Z_r$, which is the expected size of the optimal matching with expectation taken over the non-queried edges, cannot be larger than $O_r$, which is  the maximum matching assuming that every non-queried edge exists.
\end{proof}
\begin{proof}[\textsc{sketch of Theorem~\ref{thm:adaptive-adversary}}]
Let $L=\frac{4}{\epsilon}-1$. First note that for all $r$, it is true that
\begin{align*}
\label{eq:jeefa}
\expc_{Q_{r-1}}[Z_{r} - f_\delta] = \expc_{Q_{r-1}}[Z_{r}] - f_\delta & = \expc_{Q_{r-1}} \left[  \condexpmat{E}{\bigcup_{i = 1}^{r-1}Q_{i}} \right] - f_\delta \\
& = \condexpmat{E}{\bigcup_{i = 1}^{r-2}Q_{i}}-f_\delta  = Z_{r-1} - f_\delta.
\end{align*}
The remainder of the proof is similar to that of Theorem~\ref{thm:main-adaptive-matching} with $Z_r - f_\delta$ replacing $Z_r$. Following similar analysis, we have 
\[\expc_{Q_1, \dots, Q_R} [|M_R|] \geq \alpha \frac{1 - (1-\gamma)^R}{\gamma} (\expmat{E} - f_\delta).
\] 
Since $R = \frac{\log(2/\epsilon)}{\delta^{4/\epsilon}}$, we have
\begin{align}
\frac{\alpha}{\gamma} \big( 1-(1-\gamma)^R \big) \geq (1 - \frac{2}{L+ 1}) \big( 1-(1-\gamma)^R \big) \geq (1 - \frac \epsilon 2) (1- e^{-\gamma R}) \geq  (1 - \epsilon).
\end{align}
Therefore, Algorithm~\ref{alg:adaptive-matching} returns a matching with expected size of $(1 - \epsilon)(\expmat{E} - f_\delta)$.
\end{proof}

\subsection{Adaptive Algorithm in Stochastic Setting}
In this section, we consider the case where the vertex parameters are drawn independently from a distribution. 
\begin{corollary}\label{cor:adaptive-distribution}
Given any graph $(V, E)$ with vertex parameters that are drawn from distribution $D$ and any $\epsilon, \delta>0$, Algorithm~\ref{alg:adaptive-matching} returns a matching with expected size of  $(1-\epsilon)(\expmat{E} - ng_\delta)$ in $R = \frac{\log (2/\epsilon)}{\delta^{4/\epsilon}}$ iterations.
\end{corollary}
\begin{proof}
The result of Theorem~\ref{thm:adaptive-adversary} holds for any value of $f_{\delta}$. Hence, on taking expectation over the value of $f_{\delta}$, we have our result.
\end{proof}
The next corollary shows the implication of Corollary~\ref{cor:adaptive-distribution} for the uniform distribution.
\begin{corollary}\label{cor:adaptive-uniform}
For a given graph $(V, E)$ with vertex parameters that are drawn from the uniform distribution, and any $\epsilon>0$, Algorithm~\ref{alg:adaptive-matching} returns a matching with expected size of  $(1-\epsilon)(\expmat{E} - \epsilon n)$ in 
$R =  \frac{\log (2/\epsilon)}{\epsilon^{4/\epsilon}}$ iterations.
\end{corollary}
\begin{proof}
This follows from Corollary~\ref{cor:adaptive-distribution} by setting $\delta = \epsilon$ and noting that  $g_{\epsilon} = \epsilon$ for the uniform distribution. 
\end{proof}

\subsection{Non-adaptive algorithm in Adversarial Setting}
In this section, we consider the case where an adversary chooses the values of vertex parameters. We prove performance guarantees for Algorithm~\ref{alg:non-adaptive-matching} in this adversarial setting.
\begin{theorem}
\label{thm:non-adaptive-adversary}
Given a graph $(V,E)$ with vertex parameters that are selected by an adversary, and any  $\epsilon, \delta>0$, 
Algorithm~\ref{alg:non-adaptive-matching} returns a matching with expected size of $\frac{1}{2}(1-\epsilon)(\expmat{E} - f_\delta)$ in $R = \frac{\log(2/\epsilon)}{\delta^{4/\epsilon}}$ iterations.
\end{theorem}
The proof of Theorem~\ref{thm:non-adaptive-adversary} and the subsequent lemma are similar to Section~\ref{sec:non-adaptive-matching}, and are included here for completeness. 

\begin{lemma}
\label{lem:non-adaptive-induction-adversary}
For any iteration $r\in [R]$ of Algorithm~\ref{alg:non-adaptive-matching} and odd $L$, if $\expmat{W_{r-1}} \leq \expmat{E}/2$, then $\expmat{W_{r}} \geq \frac{\alpha}{2}~(\expmat{E}-f_\delta) + (1- \gamma)\expmat{W_{r-1}}$, where  $\alpha = \delta^{L+1}$ and $\gamma = \delta^{L+1}(1+\frac{2}{L+1})$.
\end{lemma}

\begin{proof}
Define $U = E \setminus  W_{r-1}$. Assume that $\expmat{W_{r-1}}\leq \expmat{E} / 2$.
By Claim~\ref{clm:exp-match-sum-parts}, we know that $\expmat{U} \ge \expmat{E} - \expmat{W_{r-1}}$. Hence, $\card{O_{r}} = \card{M(U)} \ge \expmat{U} \ge \expmat{E} - \expmat{W_{r-1}}\ge \expmat{E}/2$. 

Let $W'_{r-1}$ represent one possible outcome of existing edges when edges are drawn from $W_{r-1}$.
By Lemma~\ref{lem:struct-res-matching}, there are at least $\card{O_{r}} - (1+\frac{2}{L+1})\card{M(W'_{r-1})}$ augmenting paths of length at most $L$ in $O_{r} \Delta M(W'_{r-1})$ that augment $M(W'_{r-1})$. 
Among these paths, at most $f_\delta$  have a vertex $v_i$,  with $p_i<\delta$. We ignore these paths. 
Each remaining path succeeds with probability $(\delta^2)^{(L+1)/2}$. Hence, the expected increase in the size of $|M(W'_{r-1})|$ using the remaining paths of length $L$ is,
\begin{align*}
\condexpmat{O_{r} \cup W'_{r-1}}{W'_{r-1}} - \card{M(W'_{r-1})}& \geq \delta^{L+1} \left(\card{O_{r}} - (1+ \frac{2}{L+1})\card{M(W'_{r-1})} - f_\delta \right)\\
&\geq \delta^{L+1}~\left (\frac{1}{2}~\expmat{E}- (1+ \frac{2}{L+1})\card{M(W'_{r-1})} - f_\delta\right).
\end{align*} 
Re-arranging the inequality, we get
$\condexpmat{O_{r} \cup W'_{r-1}}{W'_{r-1}} \geq \frac{\alpha}{2}~(\expmat{E}-f_\delta) + (1- \gamma)\card{M(W'_{r-1})}$.
Taking expectation over the coin tosses on $W_{r-1}$ that create outcome $W'_{r-1}$, we have
\[
\expmat{W_r} \geq \expc_{W_{r-1}} [\condexpmat{O_{r} \cup W'_{r-1}}{W'_{r-1}}]
\geq \frac{\alpha}{2}~(\expmat{E} -f_\delta) + (1- \gamma)\expmat{W_{r-1}}.\]
\end{proof}
\begin{proof}[\textsc{sketch of Theorem~\ref{thm:non-adaptive-adversary}}]
Let $L= \frac 4\epsilon -1$.
The proof is similar to that of Theorem~\ref{thm:non-adaptive-alg} with the value of $\expmat{E}$ being replaced by $\expmat{E} - f_\delta$. Following a similar analysis, we get
\begin{align*}
\expmat{W_{R}} \geq \frac{\alpha}{2}\frac{(1 - (1-\gamma)^{R})}{\gamma}(\expmat{E}-f_\delta).
\end{align*}
Now, $\frac{\alpha}{\gamma}(1 - (1-\gamma)^{R}) \geq (1 - \frac{2}{L+1})(1- e^{-\gamma R}) \geq (1 - \epsilon)$ for $R = \frac{\log (2/ \epsilon)}{ \delta^{4/\epsilon}}$.  Hence, Algorithm~\ref{alg:non-adaptive-matching} returns a matching with expected size of  $0.5 (1-\epsilon)(\expmat{E}-f_\delta)$.
\end{proof}

\subsection{Non-adaptive algorithm in Stochastic Setting}
We examine the performance of Algorithm~\ref{alg:non-adaptive-matching} in the  setting where the vertex parameters are chosen independently from a distribution.
\begin{corollary}\label{cor:non-adaptive-distribution-2}
Given a graph $(V,E)$ with vertex parameters that are selected from distribution $D$, and $\epsilon, \delta > 0$, Algorithm~\ref{alg:non-adaptive-matching} 
returns a 
 matching with expected size of  $\frac12 (1-\epsilon) (\expmat{E} - n g_\delta)$ with  $R = \frac{\log (2/\epsilon)}{\delta^{4/\epsilon}}$ non-adaptive queries.
\end{corollary}
\begin{proof}
The result of Theorem~\ref{thm:non-adaptive-adversary} holds for any value of $f_{\delta}$. Hence, on taking expectation over the values of $f_{\delta}$, we have our result.
\end{proof}

\begin{corollary}\label{cor:non-adaptive-uniform}
For any
$G=(V, E)$ with vertex parameters that are drawn from the uniform distribution, and any $\epsilon>0$, Algorithm~\ref{alg:non-adaptive-matching}  returns a matching with expected size of  $0.5(1-\epsilon) (\expmat{E} - n\epsilon)$ with $R = \frac{\log (2/\epsilon)}{\epsilon^{4/\epsilon}}$ non-adaptive queries. 
\end{corollary}
\begin{proof}
This follows from Corollary~\ref{cor:non-adaptive-distribution-2} by setting $\delta = \epsilon$ and noting that  $g_{\epsilon} = \epsilon$ for the uniform distribution.
\end{proof}

\section{Additional experimental results on UNOS compatibility graphs}
\label{app:experiments-appendix}

In this section, we include additional experimental results on the same \num{169} compatibility graphs drawn from the real UNOS kidney exchange used in Section~\ref{sec:experiments}.  These experiments mimic those of Section~\ref{sec:experiments-unos}, only this time including in the analysis empty omniscient matchings.  If an omniscient matching is empty, then our algorithm will achieve at most zero matches as well.  In the body of this paper, we removed these cases from the experimental analysis because achieving zero matches (using any method) out of zero possible matches trivially achieves \num{100}\% of the omniscient matching; by not including those cases, we provided a more conservative experimental analysis.  In this section, we include those cases and rerun the analysis.

Figure~\ref{fig:experiments-unos2-with-zero} mimics Figure~\ref{fig:experiments-unos2-without-zero} from the body of this paper.  It shows results for \num{2}-cycle matching on the UNOS compatibility graphs, without chains (left) and with chains (right), for $R \in \{0,1,\ldots,5\}$ and varying levels of $f \in \{0,0.1,\ldots,0.9\}$.  We witness a marked increase in the fraction of omniscient matching achieved as $f$ gets close to $0.9$; this is due to the relatively sparse UNOS graphs admitting no matchings for high failure rates.

\begin{figure}[h]
\centering
\begin{minipage}[b]{0.48\textwidth}
\centering
\includegraphics[width=\textwidth]{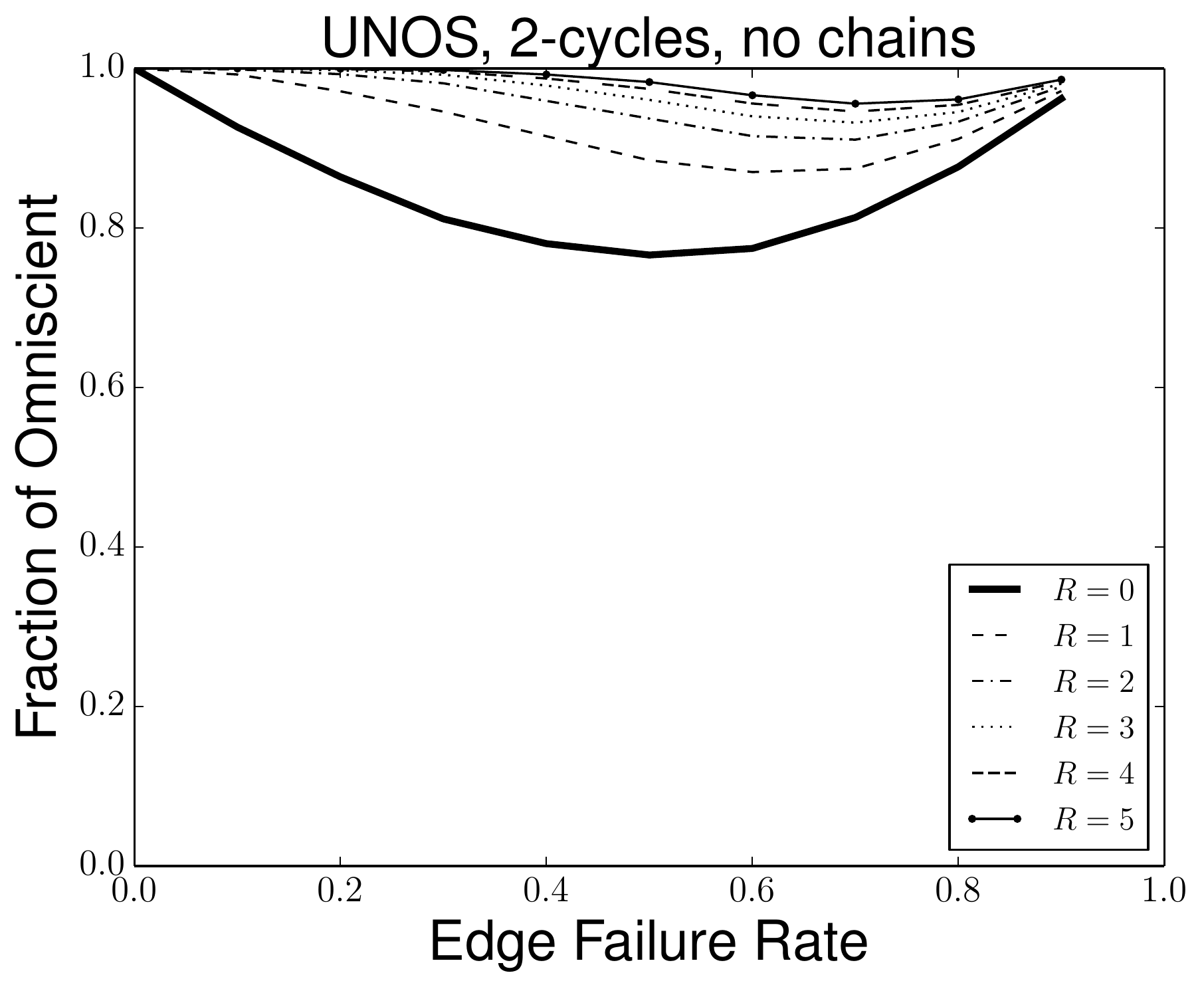}
\end{minipage}%
\hfill%
\begin{minipage}[b]{0.48\textwidth}
\centering
\includegraphics[width=\textwidth]{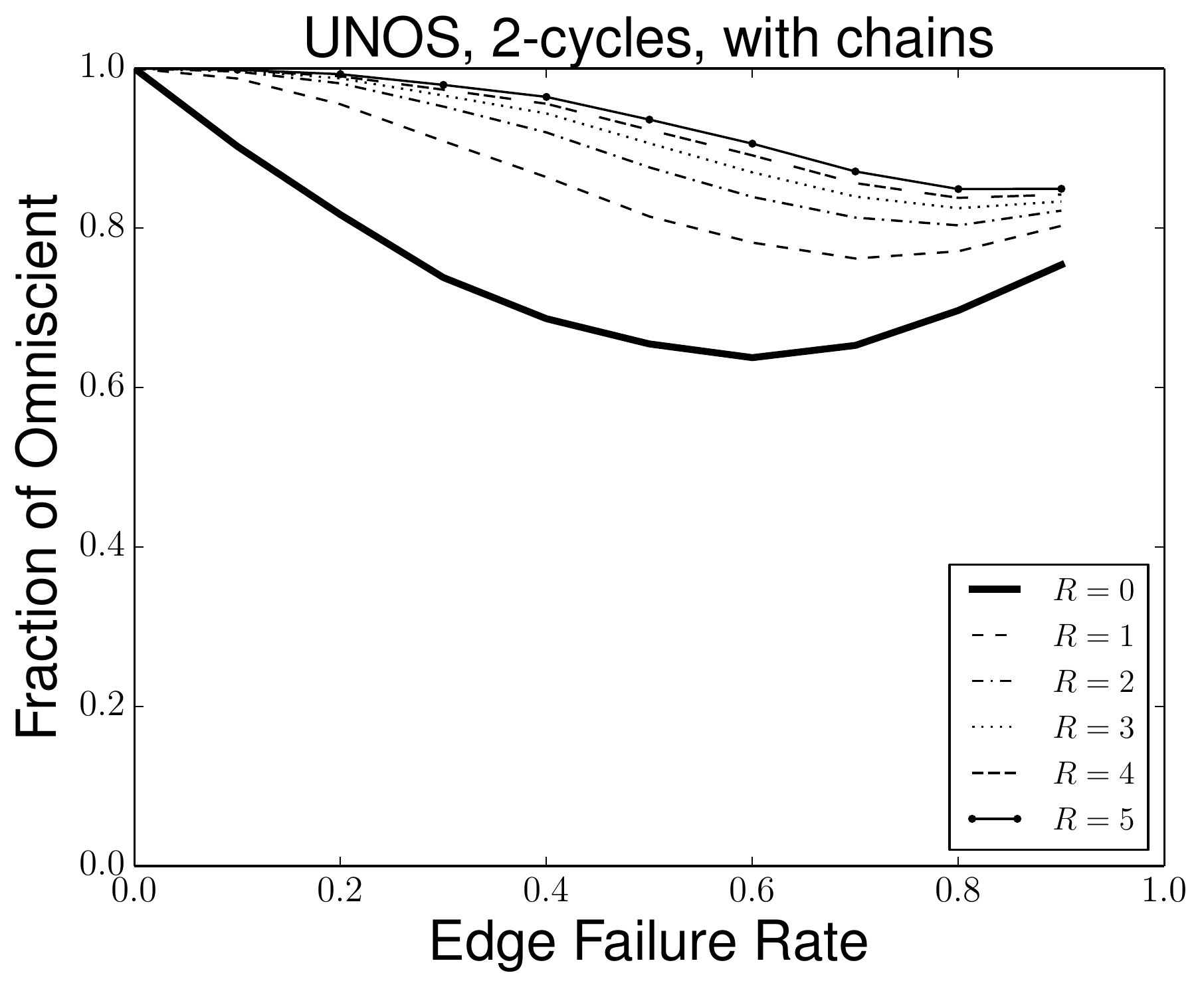}
\end{minipage}
\caption{Real UNOS match runs, restricted matching of 2-cycles only, without chains (left) and with chains (right), including zero-sized omnsicient matchings.}
\label{fig:experiments-unos2-with-zero}
\end{figure}

Figure~\ref{fig:experiments-unos3-with-zero} shows the same experiments as Figure~\ref{fig:experiments-unos2-with-zero}, only this time allowing both \num{2}- and \num{3}-cycles, without (left) and with (right) chains.  It corresponds to Figure~\ref{fig:experiments-unos3-without-zero} in the body of this paper, and exhibits similar but weaker behavior to Figure~\ref{fig:experiments-unos2-with-zero} for high failure rates.  This demonstrates the power of including \num{3}-cycles in the matching algorithm---we see that far fewer compatibility graphs admit no matchings under this less-restrictive matching policy.

\begin{figure}[h]
\centering
\begin{minipage}[b]{0.48\textwidth}
\centering
\includegraphics[width=\textwidth]{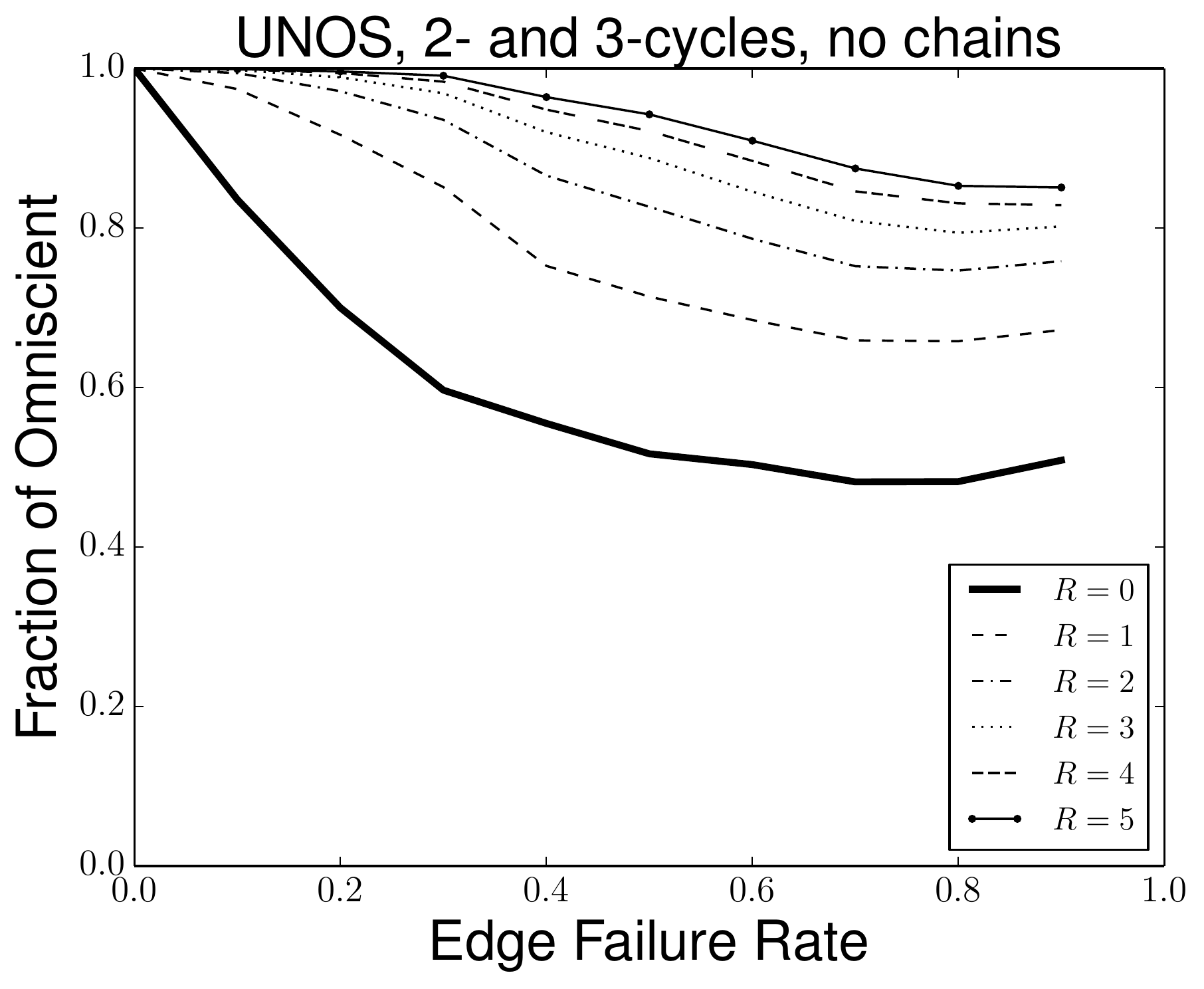}
\end{minipage}%
\hfill%
\begin{minipage}[b]{0.48\textwidth}
\centering
\includegraphics[width=\textwidth]{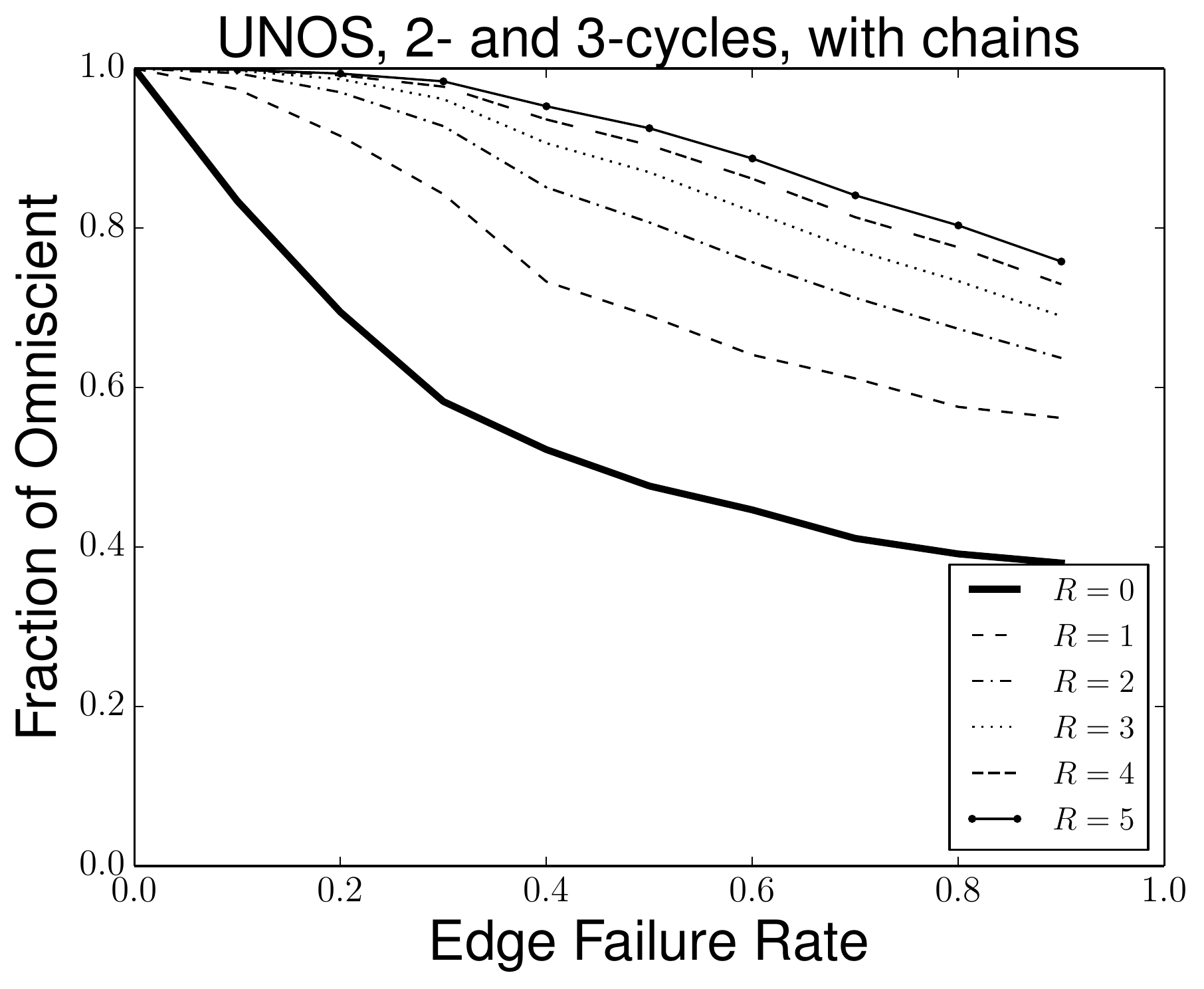}
\end{minipage}
\caption{Real UNOS match runs, matching with 2- and 3-cycles, without chains (left) and with chains (right), including zero-sized omnsicient matchings.}
\label{fig:experiments-unos3-with-zero}
\end{figure}

\end{document}